\documentclass[a4paper,11pt]{article}

\usepackage[colorlinks,
            citecolor=blue,
            linkcolor=blue,
            urlcolor=blue]{hyperref}
\usepackage{amsmath}
\usepackage{amssymb}
\usepackage{mathrsfs}
\usepackage{amsthm}
\usepackage[font=small]{caption}
\usepackage{cite}
\usepackage{color}
\usepackage[margin=1.5in]{geometry}
\usepackage{graphicx}
\usepackage{verbatim}
\usepackage{multirow}     

\newcommand{\inputtikz}[1]{\includegraphics{#1}}

\usepackage[ruled,lined,linesnumbered]{algorithm2e}
\SetKwComment{Comment}{[}{]}
\SetCommentSty{textrm}

\newcommand{\sN}{\mathbb{N}}
\newcommand{\sZ}{\mathbb{Z}}
\newcommand{\sR}{\mathbb{R}}
\newcommand{\Idk}{I^{d'}_k}
\newcommand{\Idmk}{I^{d'-1}_{2k}}
\newcommand{\Idmkp}{I^{d'-1}_{2k+1}}
\newcommand{\spmod}{\!\mod}

\newcommand{\eqnlab}[1]{\label{eqn:#1}}
\newcommand{\eqnref}[1]{\eqref{eqn:#1}}
\newcommand{\figlab}[1]{\label{fig:#1}}
\newcommand{\figref}[1]{Figure~\ref{fig:#1}}
\newcommand{\seclab}[1]{\label{sec:#1}}
\newcommand{\secref}[1]{Section~\ref{sec:#1}}
\newcommand{\thmlab}[1]{\label{thm:#1}}
\newcommand{\thmref}[1]{Theorem~\ref{thm:#1}}
\newcommand{\alglab}[1]{\label{alg:#1}}
\newcommand{\algref}[1]{Algorithm~\ref{alg:#1}}

\newcommand{\mytabvspace}{\vphantom{${X^X}^X$}} 
\DeclareMathOperator{\type}{type}

\graphicspath{{figures/}}

\newtheorem{thm}{Theorem}
\newtheorem{prop}[thm]{Proposition}
\newtheorem{cor}[thm]{Corollary}
\newtheorem{dfn}[thm]{Definition}
\newtheorem{lem}[thm]{Lemma}
\newtheorem{conj}[thm]{Conjecture}
\newtheorem{rem}[thm]{Remark}

\author{Carsten Burstedde\footnote{%
          Institut f\"ur Numerische Simulation (INS)
          and Hausdorff Center for Mathematics (HCM),
          Universit\"at Bonn, Germany,
          \href{mailto:burstedde@ins.uni-bonn.de}{burstedde@ins.uni-bonn.de}
          (corresponding author)%
        },
        Johannes Holke\footnote{%
          Institut f\"ur Numerische Simulation (INS),
          Hausdorff Center for Mathematics (HCM),
          and Bonn International Graduate School for Mathematics (BIGS),
          Universit\"at Bonn, Germany%
        },
        Tobin Isaac\footnote{%
          Computing Institute, The University of Chicago, USA
        }
      }
\title{Bounds on the number of discontinuities\\
       of Morton-type space-filling curves}

\begin{document}

\maketitle

\begin{abstract}
The Morton- or $z$-curve is one example for a space filling curve: Given a
level of refinement $L \in \sN_0$, it maps the interval $[0, 2^{dL}) \cap \sZ$
one-to-one to a set of $d$-dimensional cubes of edge length $2^{-L}$ that form
a
subdivision of the unit cube.
Similar curves have been proposed for triangular and tetrahedral unit domains.
In contrast to the Hilbert curve that is continuous, the Morton-type curves
produce jumps.

We prove that any contiguous subinterval of the curve divides the domain into a
bounded number of face-connected subdomains.
For the hypercube case in arbitrary dimension, the subdomains are star-shaped
and the bound is indeed two.
For the simplicial case in dimensions 2 and 3, the bound is proportional to the
depth of refinement $L$.
We supplement the paper with theoretical and computational studies on the
frequency of jumps for a quantitative assessment.
%
\end{abstract}

\section{Introduction}
\seclab{intro}

\begin{figure}\centering
  \begin{minipage}{0.49\textwidth}\centering
    \inputtikz{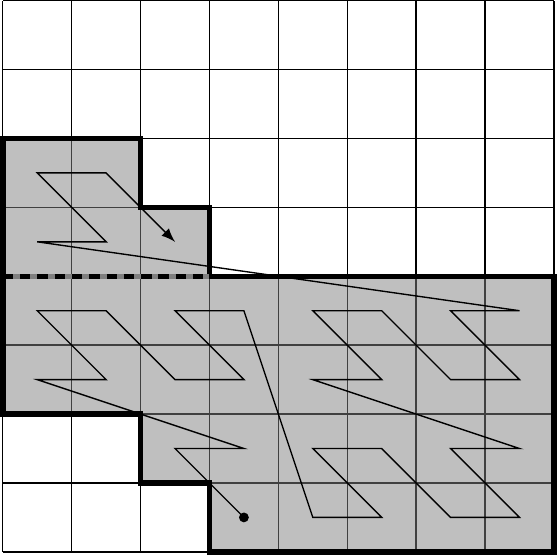}
  \end{minipage}
  \begin{minipage}{0.49\textwidth}\centering
    \inputtikz{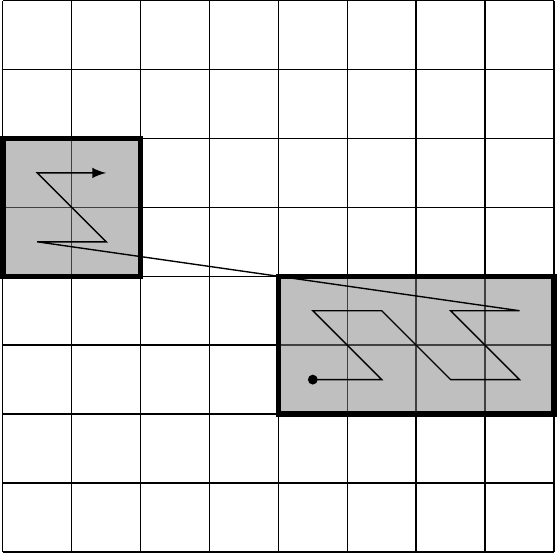}
  \end{minipage}
  \caption{Contiguous subsections of the Morton curve at refinement
           level $L = 3$.  Observe the jumps when the $z$-curve
           runs diagonally.  In the right hand image, this
           produces two disconnected subdomains.
           In both pictures shown, the domain decomposes into two star-shaped
           pieces.}
  \figlab{morton}
\end{figure}
The Peano curve \cite{Peano90} and the Hilbert curve \cite{Hilbert91} are
continuous maps from the line onto the $d$-dimensional unit cube.  A large
number of such space filling curves (SFC) has been described in the literature; see
for example \cite{Sagan94, Bader12, HaverkortWalderveen10} and the references
therein.  They are usually defined in terms of a recursive prescription.  For
numerical applications, the curve is made discrete and finite by bounding the
depth of the recursion.  The smallest units of space that are traversed may be
called elements.  The Morton- or $z$-curve, originally described by Lebesgue
\cite{Lebesgue04} and adapted to data storage by Morton \cite{Morton66}, also
creates such a map, but it is not continuous.  In fact, it contains jumps
throughout its length (see \figref{morton}).  This raises the concern that a
subsection of the curve may divide the space covered by its image into a large
number of disconnected subdomains.  Especially when the curve is used to divide
a computational mesh between different processors for parallel computation, see
e.g.\ \cite{GriebelZumbusch99, AkcelikBielakBirosEtAl03,
BursteddeGhattasGurnisEtAl10, AhimianLashukVeerapaneniEtAl10,
WeinzierlMehl11},
the surface-to-volume ratio of a fragmented subdomain could grow without
bounds, which would likely increase the amount of data to be communicated.

In this paper we eliminate that concern by proving that the (classic cubical)
Morton curve can lead to no more than two subdomains, where we define a set of
elements to be of the same subdomain if they are connected by a finite number
of element face connections:
\begin{thm}
  \label{illthmallcube}
  A contiguous segment of a Morton curve through a uniform or adaptive tree
  of maximum refinement level $L$ produces at most two distinct face-connected
  subdomains.
  This result is independent of the space dimension.
\end{thm}

We note that a proof for the two-dimensional square has been given in
\cite[pages 175--177]{Bader12} that proceeds by illustrating and enumerating a
finite number of cases.
(In fact, we adapt these ideas to dimensions two and three in
\secref{illustrated}, and also restate the extension to adaptive meshes in
\secref{main}.)
It is said in \cite{Bader12} that the construction extends to dimensions three
and higher.
This is entirely plausible, yet we see that the number of cases to discuss
grows with the space dimension and would eventually require some kind of
automation.
Thus, we proceed inductively over $d$ to provide dimension-independent
results.
%
We also supply a formal non-inductive proof to show that the
connected segments are star-shaped.

For the triangular and tetrahedral Morton curves introduced recently
\cite{BursteddeHolke16}, we show that the bound is proportional to the depth of
refinement $L$:
\begin{thm}
  \label{illthmalltets}
  A contiguous segment of a tetrahedral Morton curve through a uniform or
  adaptive tree of maximum refinement level $L\geq 2$ produces at most $2(L-1)$
  face-connected subdomains in 2D and at most $2L+1$ in 3D.
  For $L=1$ there are at most two face-connected subdomains.
\end{thm}

We complete our study with a statement on the lower bound on the fraction of
continuous segments and provide an algorithm and numerical results to
illustrate the distribution of continuous vs.\ discontinuous segments.
This supports our conjecture that the tetrahedral Morton curve is no worse in
practice than the original cubical construction.


%

\section{Concepts and notation}
\seclab{concepts}

There is a natural identification between Morton-ordered elements on the one
hand and uniform and adaptive quadtrees \cite{FinkelBentley74} and octrees
\cite{Meagher82} on the other.
This is true for the tetrahedral Morton curve \cite{BursteddeHolke16} as well.
We will often refer to the elements as (sub)quadrants irrespective of the shape
or space dimension $d$.
Different ways exist to formalize the definition of a general space filling
curve; one is to identify a finite set of types of transformations and rules to
apply them recursively \cite{HaverkortWalderveen10}.
In this document we restrict the theory and notation to the minimum required to
treat the cubical and the tetrahedral Morton curve.

\subsection{The cubical Morton curve}
\seclab{concepts-cubical}

The Morton subdivision of a $d$-dimensional hypercube \cite{Morton66} can be
constructed by recursion.  When dividing a cube into $2^d$ half-size subcubes,
we enumerate these with the binary index
\begin{equation}
  \eqnlab{onelevel}
  q = (q_d \ldots q_1)_2 \in [0, 2^d) \cap \sZ
\end{equation}
comprised of $d$ bits $q_i \in \{ 0, 1 \}$.  (We will drop $\cap \sZ$ in the
following when it is clear that we are referring to integers.)
Each of the bits $i$ corresponds to the position of that subcube in the $x_i$
coordinate direction, where 0 denotes the lower and 1 the higher half.  In
our convention the most significant bit corresponds to the last dimension ($z$
in three dimensions) and the least significant bit to $x \equiv x_1$.
When counting through the possible values of $q$ we see that the $x_1$
coordinate changes its value fastest and the $x_d$ coordinate slowest.  Before
a bit at position $i$ flips, all numbers in in the lower $i-1$ bits have to be
counted through first.

We can state one central and well known fact at this point: The flip of the
$i$th bit amounts to a shift of the corresponding subcube parallel to the
coordinate direction $i$.  If we flip from zero to one, we move up, and else we
move down the axis.  It is easy to see that flipping one bit transforms the
subcube into its neighbor across a face with normal direction $\pm x_i$.

We define a recursion by subdividing each subcube further using the same
prescription.  The root cube is associated with level $\ell = 0$, with levels
increasing with each subdivision.  Subcubes exist at any level $\ell$ and are
identified with the root of a corresponding subtree.  Level-$L$ subtrees are
also called subquadrants.  We count the sequence of level $L$ (sub)quadrants
with the index
\begin{equation}
  \eqnlab{levels}
  Q = (q^1 \ldots q^L)_2 \in [0, 2^{dL})
  ,
\end{equation}
where each level-wise index $q^\ell$ is defined as in \eqnref{onelevel}.
They designate the choice of subquadrants from the first subdivision
$\ell = 1$ to the last at level $\ell = L$.  This sequence of choices can be
understood as the path from the root to the leaf of a decision tree, where each
decision is between $2^d$ possibilities.  The subset of $\mathbb{R}^d$
occupied by the quadrant with index $Q$ is
\begin{equation}\eqnlab{rdset}
  \begin{aligned}
    \Omega(Q) := &[2^{-L}(q^1_1 q^2_1 \dots q^L_1)_2,
     2^{-L}((q^1_1 q^2_1 \dots q^L_1)_2 + 1)] \times \\
    &[2^{-L}(q^1_2 q^2_2 \dots q^L_2)_2,
     2^{-L}((q^1_2 q^2_2 \dots q^L_2)_2 + 1)] \times \\
    &\vdots \\
    &[2^{-L}(q^1_d q^2_d \dots q^L_d)_2,
     2^{-L}((q^1_d q^2_d \dots q^L_d)_2 + 1)].
   \end{aligned}
\end{equation}

We define a full or complete subtree by the set of all its descendant
quadrants.  A subtree is incomplete if the quadrants form a strict subset of
descendants that are contiguous with respect to the indexing \eqnref{levels}.
We call such a subset a segment of a Morton curve in the following (two
examples are depicted in \figref{morton}).

We will make use of the following symmetry property of the Morton curve: It can
be traversed forward or in reverse.  The reversal amounts to go through the
indexing \eqnref{levels} by counting backwards.  A quadrant is transformed into
the reverse ordering by taking the bitwise negation (the one-complement) of its
index,
\begin{equation}
  \eqnlab{Rdef}
  R(Q) = 2^{dL} - 1 - Q.
\end{equation}
Geometrically, this operation mirrors the quadrant around the center point of
the root cube.

\subsection{The simplicial Morton curve}
\seclab{concepts-simplicial}

The tetrahedral Morton (TM) SFC applies to triangular and tetrahedral
red-refinement of a mesh (and, conceptually, to higher dimensional simplices)
\cite{BursteddeHolke16}.
We encounter 1:4 refinement in 2D and 1:8 refinement in 3D \cite{Bey92},
which means that the quad-/octree interpretation is still valid.
We compute the SFC in a bitwise fashion that is an extension of the
traditional Morton curve.
In order to define the TM SFC we introduce the concept of the type
of a simplex.

\begin{figure}
\center
\begin{minipage}{0.48\textwidth}
   \def\svgwidth{40ex}
   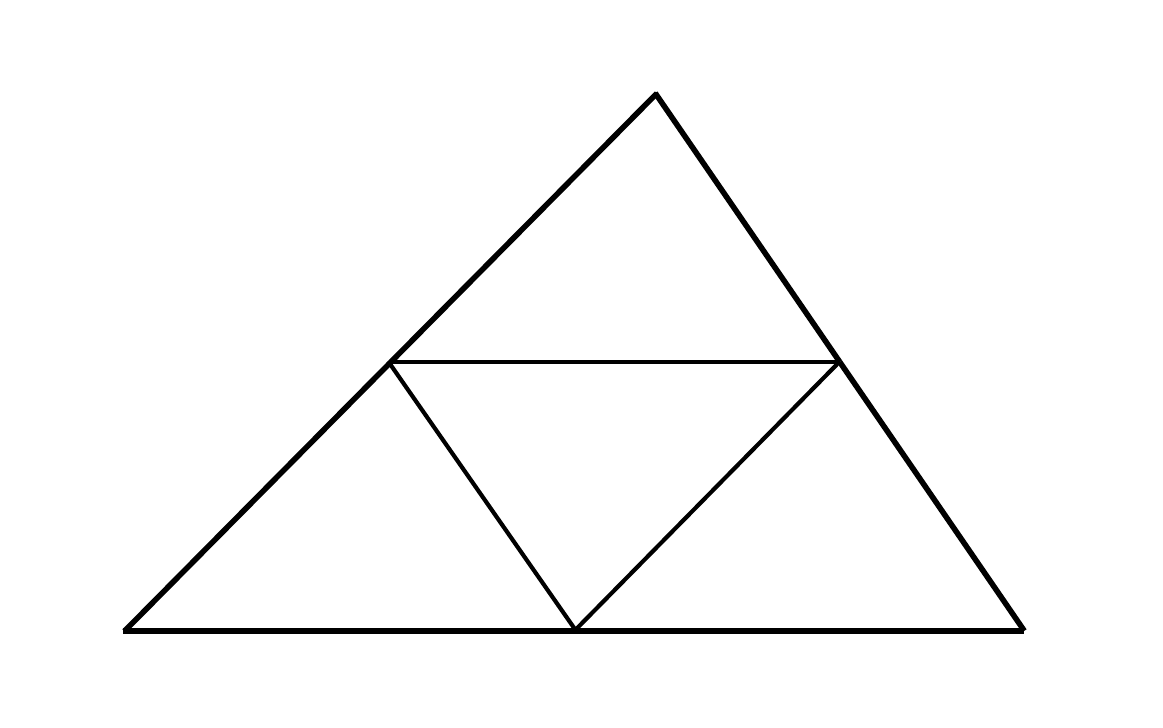
\end{minipage}
\begin{minipage}{0.48\textwidth}
   \def\svgwidth{40ex}
   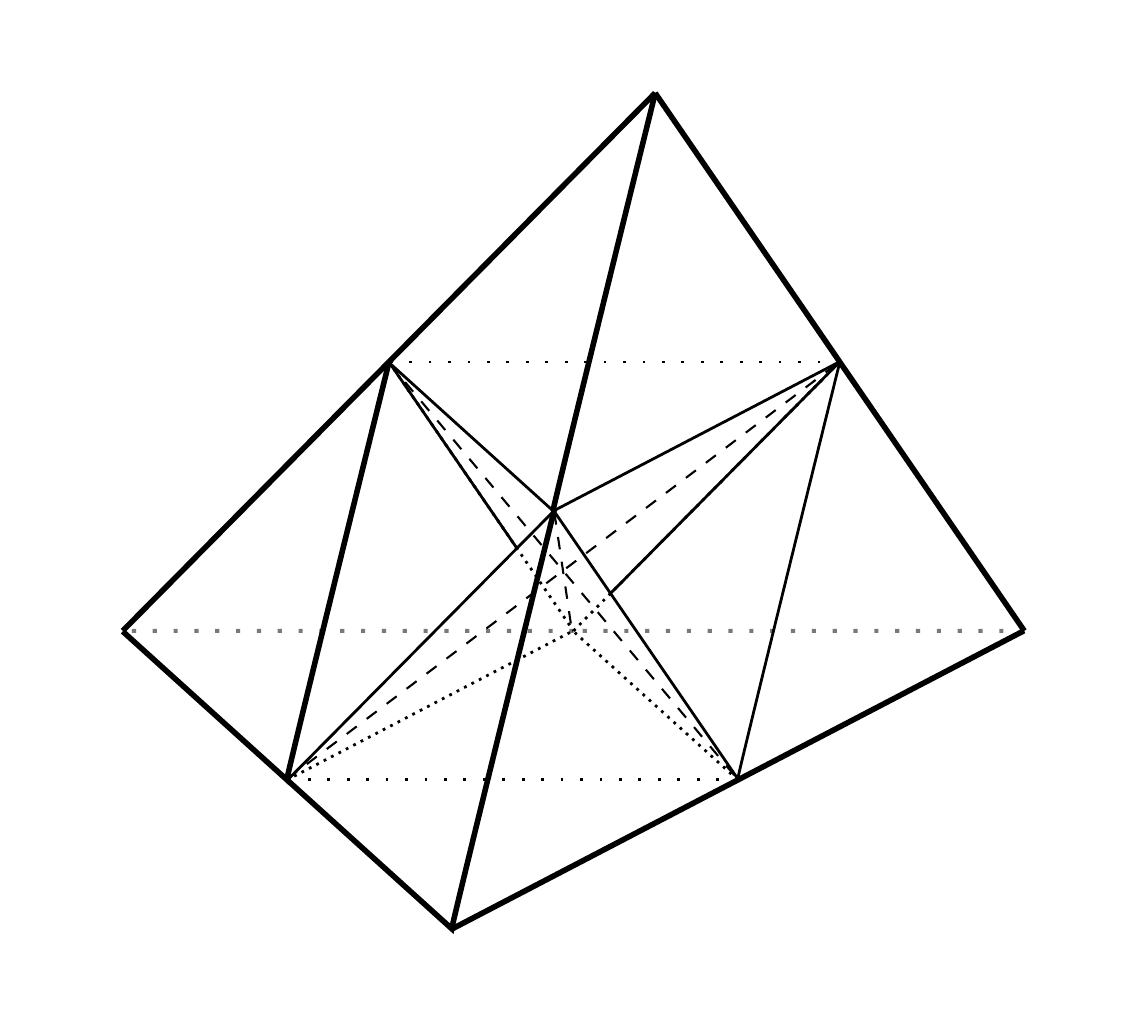
\end{minipage}
   \caption{%
   Left: the refinement scheme for triangles in two dimensions.  A triangle
   $T=[x_0,x_1,x_2]\subset\sR^2$ is refined by dividing each face at the midpoints
   $x_{ij}$. We obtain four smaller triangles, all similar to $T$.
   Right: the
   situation in three dimensions.  If we divide the edges of the tetrahedron
   $T=[x_0,x_1,x_2,x_3]\subset\sR^3$ in half, we get four smaller tetrahedra
   (similar to $T$) and one octahedron.  By dividing the inner octahedron along any
   of its three diagonals (shown dashed) we finally end up with a partition of
   $T$ into eight smaller tetrahedra, all having the same volume.  The refinement
   rule of Bey is obtained by always choosing the diagonal from $x_{02}$ to
   $x_{13}$ and numbering the corners of the children according to
   \eqref{eq:childnumbers}.}%
   \label{fig:cutoff}%
\end{figure}%

\begin{dfn}
  We describe a $d$-dimensional simplex $T\subset\sR^d$ by $d+1$ ordered
  vertices $x_0,\dots,x_d\in\sR^d$
  and write $T=[x_0,\dots,x_d]$.
  By $x_{ij}$ we denote the midpoint between $x_i$ and $x_j$.

  Bey's \emph{red-refinement} rule \cite{Bey92} for a triangle ($d=2$) or
  tetrahedron ($d=3$) amounts to dividing the parent simplex
  $T=[x_0,\dots,x_d]$ into $2^d$ subsimplices that are defined and enumerated
  as follows (see also Figure~\ref{fig:cutoff}):
 \begin{subequations}
\label{eq:childnumbers}
 \begin{equation}
 \label{eq:childnumbers_2d}
 d=2:\quad
 \begin{array}{cccccl}
  T_0&:=&[x_0,x_{01},x_{02}], &  T_1&:=&[x_{01},x_{1},x_{12}],\\
  T_2&:=&[x_{02},x_{12},x_{2}], & T_3&:=&[x_{01},x_{02},x_{12}],
  \end{array}
 \end{equation}
 \begin{equation}
\label{eq:childnumbers3d}
d=3:\quad
 \begin{array}{cccccc}
 T_0 &:=& [x_0,x_{01},x_{02},x_{03}],   & T_4 &:=& [x_{01},x_{02},x_{03},x_{13}],\\
 T_1 &:=& [x_{01},x_{1},x_{12},x_{13}], & T_5 &:=& [x_{01},x_{02},x_{12},x_{13}],\\
 T_2 &:=& [x_{02},x_{12},x_{2},x_{23}], & T_6 &:=& [x_{02},x_{03},x_{13},x_{23}],\\
 T_3 &:=& [x_{03},x_{13},x_{23},x_{3}], & T_7 &:=& [x_{02},x_{12},x_{13},x_{23}].

\end{array}
 \end{equation}
 \end{subequations}
\end{dfn}


\begin{dfn}[Type of a simplex]
  We start with a unit square/cube divided as in Figure~\ref{fig:typeofsimplex}
  and pick any of the triangles/tetrahedra as root simplex for refinement.
  Each subsimplex in a uniform level $L$ refinement of this root simplex
  is contained in a subsquare/subcube of level $L$ and is exactly one of the
  two (2D) or six (3D) simplices from Figure~\ref{fig:typeofsimplex}.
  It thus has a unique number, which we define as the \emph{type} of the simplex.
\end{dfn}

\begin{figure}
\def\svgwidth{0.4\textwidth}
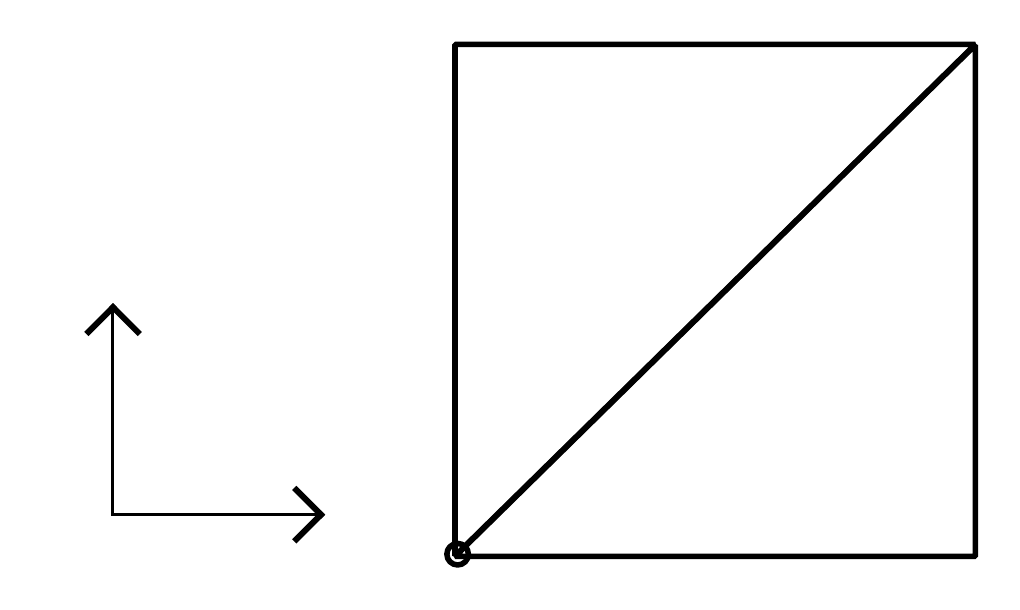
\def\svgwidth{0.58\textwidth}
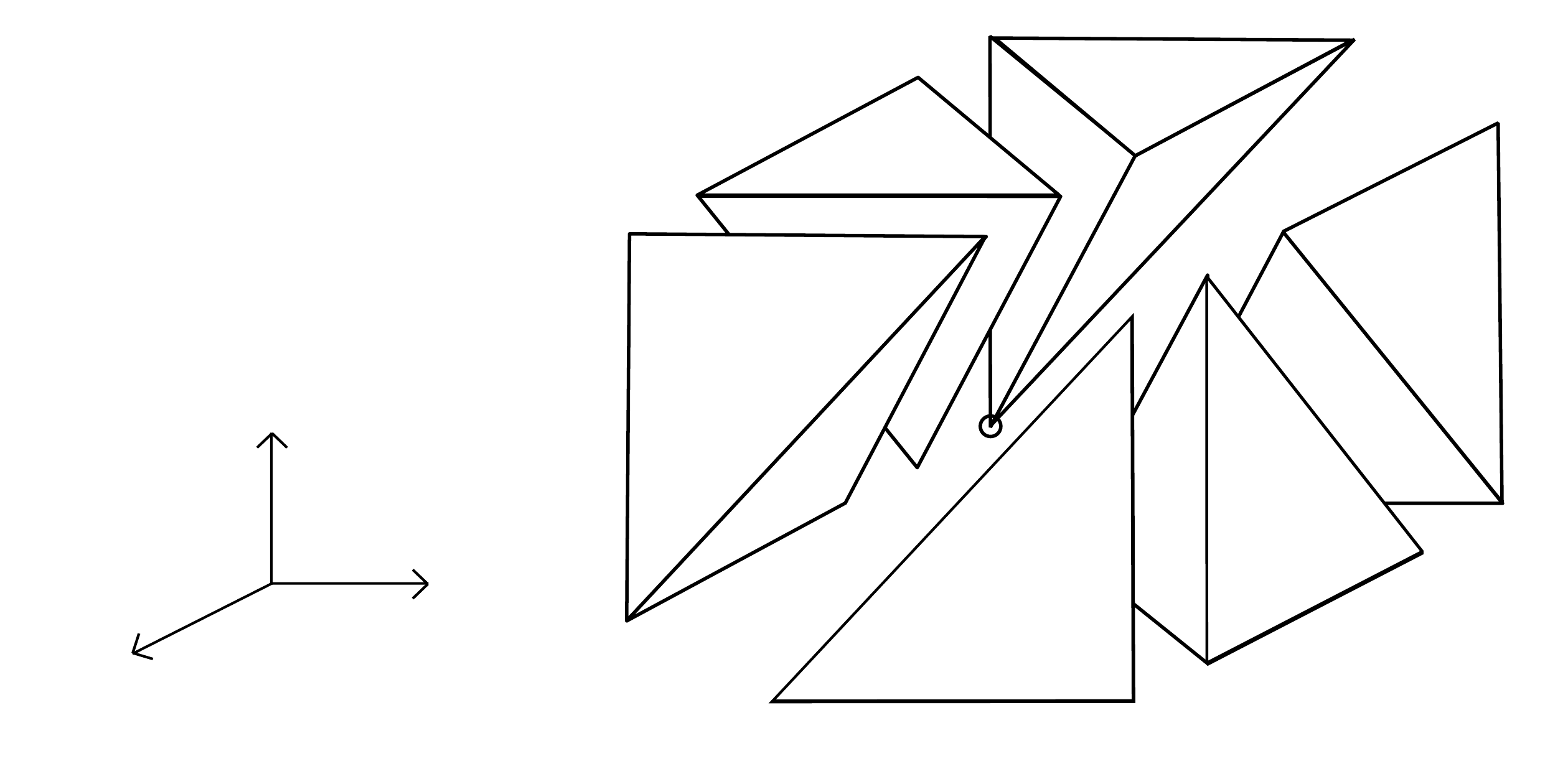
\caption{We can separate an axis aligned cube into subsimplices by dividing it
         along one diagonal. We enumerate the resulting subsimplices and call
         the number of a subsimplex its type.
         Left: We divide a 2D square into two triangles, the lower right one
               has type zero and the upper left one has type 1.
         Right: We divide a 3D cube into six tetrahedra and enumerate them
               counterclockwise from zero to five (exploded view).}
\label{fig:typeofsimplex}
\end{figure}

We start on level 0 with the root simplex $T_d^0$, which can have any of the
possible types.
In our implementation we pick 0 as the type of the root simplex.
The TM code $m(T)$ for a descendant $T$ of the root simplex is the interleaving
of its anchor (lower left) node coordinates with the types of all of $T$'s
ancestor simplices \cite{BursteddeHolke16}.
It creates a total order between all simplices of a given level and thus
establishes the SFC.
Here we give a second, recursive definition of the SFC that is more suitable
for our purposes.

By Proposition~18 in \cite{BursteddeHolke16} we obtain one permutation
$\sigma_b \in \Sigma_{2^d}$ for each possible simplex type $b$.
It relates the ordering of
its
children to the SFC
such
that for any $d$-simplex $T$ with $\type(T) = b$
\begin{equation}
 m(T_{\sigma_b(0)})<m({T_{\sigma_b(1)}})<\dots<m({T_{\sigma_b(2^d-1)}})
 .
\end{equation}
It places the child $T_i$ in Bey's
order at SFC position $\sigma_b(i)$.
\begin{dfn}
 \label{def:localindex}
 Let $T$ be a level $L$ descendant of $T_d^0$ such that $T's$ parent $P$ has
 type $b$ and $T$ is the $i$-th child of $P$ according to Bey's order
 \eqref{eq:childnumbers},
  $0\leq i<2^d$. We call the number $\sigma_b(i)$ the \emph{local index} of
 the $d$-simplex $T$ and use the notation
 \begin{equation}
I_\mathrm{loc}(T):=\sigma_b(i).
 \end{equation}
 By definition, the local index of the root simplex is zero, $I_\mathrm{loc}(T^0_d):=0$.
 Table~\ref{table:BeytoIndex} lists the local indices for each parent type.
\end{dfn}
Thus, we know for each type $0\leq b<d!$ how the children of a tetrahedron of type $b$ are traversed.
This gives us an approach for describing the SFC arising from the TM-index in a recursive fashion \cite{HaverkortWalderveen10}.
By specifying for each possible type $b$ the order and types of the children of a type $b$ simplex, we can build up the SFC.
In Figure~\ref{fig:haverkortSFC} 
we describe the SFC for triangles in this way.
In three dimensions it is not convenient to draw the six pictures for the
different types, yet
the SFC can be derived similarly from \eqref{eq:childnumbers} and Table~\ref{table:BeytoIndex}.

\begin{rem}
\label{rem:simplexsym}
  In 2D, we will make use of a symmetry property similar to \eqnref{Rdef}:
  Reversing the TM curve in a uniform refinement of a type 0 triangle
  results in the (forward) TM curve for a type 1 triangle, and vice versa.
\end{rem}

\begin{table}
\centering
\raisebox{6ex}{
\begin{tabular}{|rc|cccl|}
\hline
\multicolumn{2}{|c|}{ \mytabvspace$I_\mathrm{loc}$}&\multicolumn{4}{c|}{Child}\\
 \multicolumn{2}{|c|}{2D}  &\mytabvspace  $T_0$ & $T_1$ & $T_2$ & $T_3$ \\[0.2ex]\hline
 \multirow{2}{*}{b}&\mytabvspace0 & 0 & 1 & 3 & 2    \\[0.2ex]
 &1 & 0 & 2 & 3 & 1   \\ \hline
\end{tabular}
}
\begin{tabular}{|rc|cccccccl|}
\hline
\multicolumn{2}{|c|}{\mytabvspace$I_\mathrm{loc}$}&\multicolumn{8}{c|}{Child}\\
\multicolumn{2}{|c|}{3D}&
\mytabvspace   $T_0$ & $T_1$ & $T_2$ & $T_3$ & $T_4$ & $T_5$ & $T_6$ & $T_7$\\[0.2ex]\hline
 \multirow{6}{*}{b}&\mytabvspace0 & 0 & 1 & 4 & 7 & 2 & 3 & 6 & 5   \\[0.2ex]
 &1 & 0 & 1 & 5 & 7 & 2 & 3 & 6 & 4   \\[0.2ex]
 &2 & 0 & 3 & 4 & 7 & 1 & 2 & 6 & 5   \\[0.2ex]
 &3 & 0 & 1 & 6 & 7 & 2 & 3 & 4 & 5   \\[0.2ex]
 &4 & 0 & 3 & 5 & 7 & 1 & 2 & 4 & 6   \\[0.2ex]
 &5 & 0 & 3 & 6 & 7 & 2 & 1 & 4 & 5   \\ \hline
\end{tabular}
\caption{The local index of a $d$-simplex $T$. For each $b = \type(T)$, the
    $2^d$ children $T_0,\dots,T_{2^d-1}$ of $T$ can be ordered
    according to their TM-indices.
    The position of the $i$-th child according to this order is the local index $I_\mathrm{loc}(T_i)$.}  
\label{table:BeytoIndex}
\end{table}

\begin{figure}
   \center
   \def\svgwidth{0.9\textwidth}
   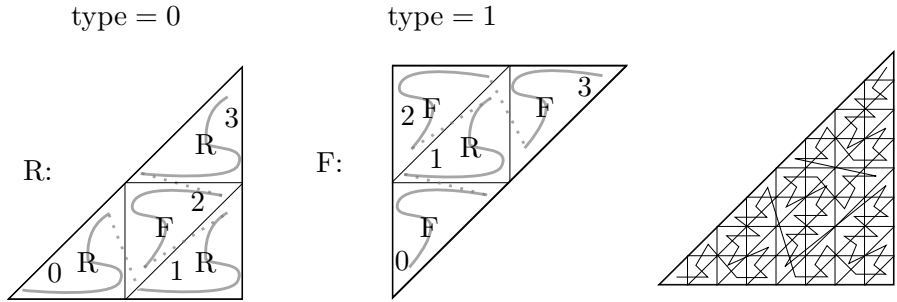
    \caption{Left: Using the notation from \cite{HaverkortWalderveen10} we
    recursively describe the space-filling curve arising from the TM-index for
    triangles. The numbers inside the child triangles $T_i$
    are their local indices $I_\mathrm{loc}(T_i)$.
    We write $R$ for the refinement scheme of type 0 triangles and $F$ for type
    1 triangles.
    This pattern can be obtained from \eqref{eq:childnumbers} and
    Table~\ref{table:BeytoIndex}.
    Right: the SFC for a uniform level 3 refinement of the root triangle.}
    \label{fig:haverkortSFC}
\end{figure}%

\section{Illustrated proofs for $d \le 3$}
\seclab{illustrated}

This section is devoted to proofs that use geometric intuition in two and
three dimensions.
For the cubical Morton curve, the idea is not new (although the execution in 3D
seems to be).
For the tetrahedral Morton curve, this is the first such study as far as we
know.
For abstract proofs for cubes of arbitrary dimension $d$ we refer the reader to
\secref{arbitrary}.

\subsection{The cubical case}
\seclab{illustrated-cubical}

In this section we prove a set of statements for cubes up to three dimensions
by providing selected illustrations and covering all possible cases.
A similar argument has been explored before in two dimensions \cite{Bader12},
while an abstract proof for two dimensions can be found in
\cite{deBergHaverkortThiteEtAl10}.

We begin with statements that assume a curve that either begins with the first
subquadrant of the unit cube or ends with its last subquadrant.  In a second
step, we use these statements to prove the final result.  All statements are
stated for arbitrary levels of refinement $L \ge 0$.  In fact, all statements
are trivially true for one dimension $d=1$ (with no jumps at all); in this
section we cover $d = 2$ and $d = 3$.
\begin{figure}\centering
  \begin{minipage}{0.49\textwidth}\centering
    \inputtikz{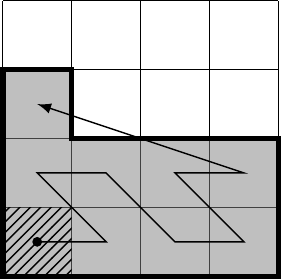}
  \end{minipage}
  \begin{minipage}{0.49\textwidth}\centering
    \inputtikz{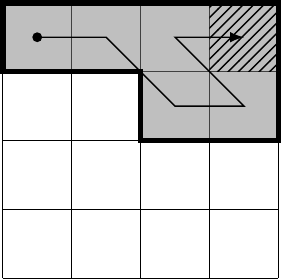}
  \end{minipage}
  \caption{We show two $L = 2$ examples of Morton curve segments that begin
           with the first subquadrant of a tree (left) and end with its last
           subquadrant (right), respectively.  In both cases the segment covers
           one face-connected subdomain.}
  \figlab{onepiece}
\end{figure}
\begin{prop}
  \label{illpropfirst}
  In a quadtree (or octree) $T$ that is uniformly refined to level $L$, a
  contiguous segment of a Morton curve that begins with the first subquadrant
  in $T$ creates exactly one subdomain of face-connected quadrants, no matter
  where it ends.
\end{prop}
\begin{cor}
  \label{illcorlast}
  In the situation of Proposition~\ref{illpropfirst}, a contiguous segment that ends
  with the last subquadrant in $T$ creates exactly one face-connected
  subdomain, no matter where it begins (see \figref{onepiece} for an
  illustration).
\end{cor}
\begin{proof}
  Assuming that Proposition~\ref{illpropfirst} is true, we can use the symmetry of
  the $z$-curve with respect to reversal to transform the present
  problem into the setting covered in Proposition~\ref{illpropfirst}.
\end{proof}
\begin{figure}\centering
  \begin{minipage}{0.32\textwidth}\centering
    \inputtikz{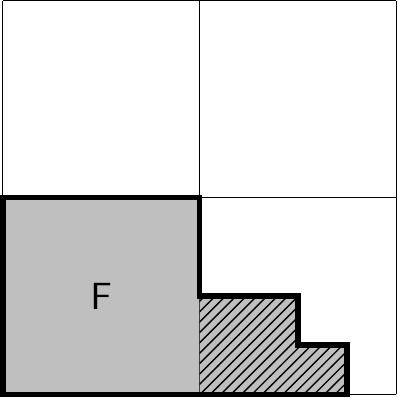}
  \end{minipage}
  \begin{minipage}{0.32\textwidth}\centering
    \inputtikz{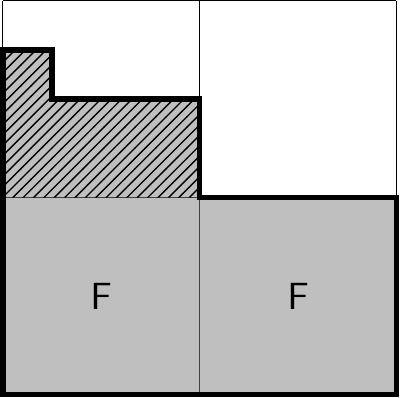}
  \end{minipage}
  \begin{minipage}{0.32\textwidth}\centering
    \inputtikz{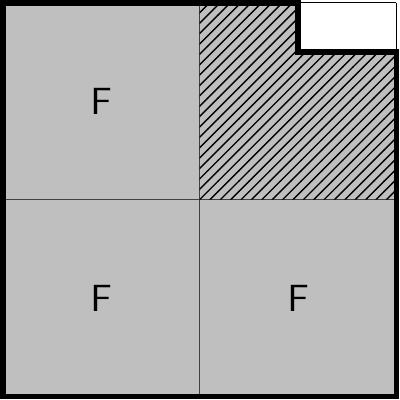}
  \end{minipage}
  \caption{Proof of Proposition~\ref{illpropfirst}:  The three non-trivial cases
    that occur in two dimensions (we choose $L = 3$).  The letter $F$
    designates a fully covered subtree of level $L - 1$.  It is crucial that
    the lower left corner of the hatched area touches at least one of the full
    subtrees across a face.}
  \figlab{twodproof}
\end{figure}
\begin{figure}\centering
  \begin{minipage}{0.49\textwidth}\centering
    \inputtikz{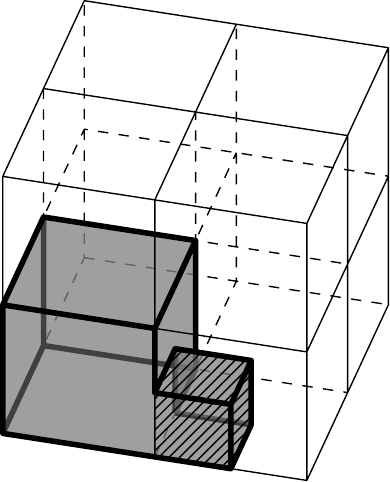}
  \end{minipage}
  \begin{minipage}{0.49\textwidth}\centering
    \inputtikz{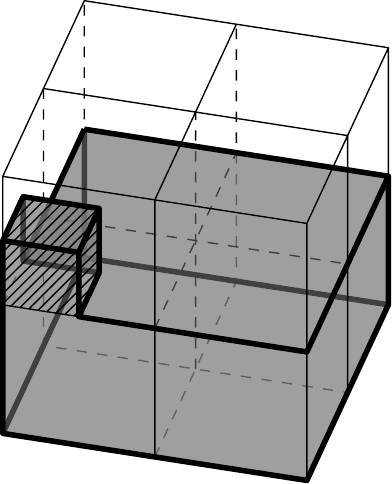}
  \end{minipage}
  \caption{Proof of Proposition~\ref{illpropfirst}:   Selected cases in three
    dimensions ($j = 1, 4$ out of the seven non-trivial ones).  The full
    subtrees are shaded lightly.  Again we exploit the fact that the lower
    left front corner of the last non-empty subtree (hatched) connects to at
    least one full subtree with a lower subtree index across a face.}
  \figlab{threedproof}
\end{figure}
\begin{proof}[Proof of Proposition~\ref{illpropfirst}]
  We proceed by induction over $L$.  Starting with $L=0$, we only have one
  element and the statement is true.  Supposing $L > 0$, we can identify the
  number $j \in [0, 2^d)$ that designates in which level 1
  subquadrant of the tree the last level $L$ subquadrant of the segment lies.
  If $j = 0$ then the whole segment is contained in a level $L-1$ subtree and
  we can apply the induction assumption.  Each of the remaining cases produces
  $j$ full subtrees and one possibly incomplete one.  That last subtree
  necessarily contains its first level $L$ subquadrant $q$.  Since this subtree
  produces one subdomain by induction, we are done by arguing that the full
  subtrees are face-connected to each other and to $q$, directly or indirectly.
  For two dimensions we show the three possible cases in \figref{twodproof},
  all of which satisfy the statement.  For three dimensions we proceed by
  enumeration as well; we show selected situations in \figref{threedproof} to
  conclude the proof.
\end{proof}
Now that we have identified situations that produce one subdomain only, we can
prove the main statement for arbitrary segments by a divide-and-conquer
approach.
\begin{prop}
  \label{illpropboth}
  In a quadtree or octree that is uniformly refined to level $L$, a contiguous
  segment of the Morton curve creates no more than two distinct face-connected
  subdomains.
\end{prop}
\begin{figure}\centering
  \begin{minipage}{0.49\textwidth}\centering
    \inputtikz{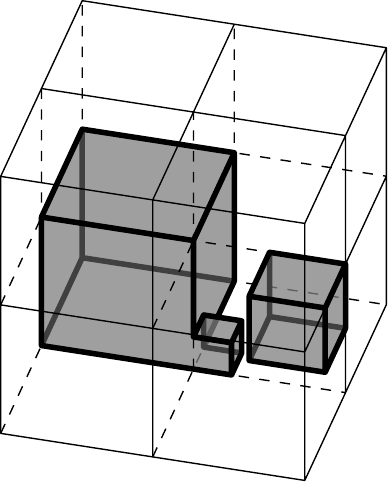}
  \end{minipage}
  \begin{minipage}{0.49\textwidth}\centering
    \inputtikz{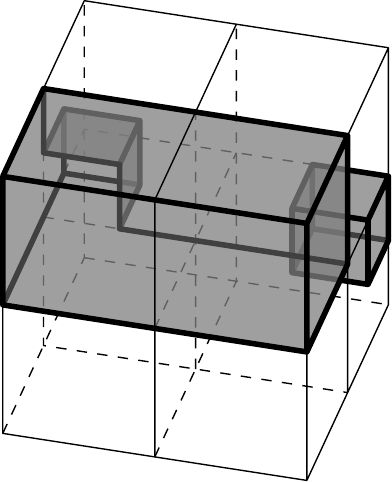}
  \end{minipage}
  \caption{The cases 1--3 and 3--6{} in the proof of
           Proposition~\ref{illpropboth} for $d = 3$ dimensions.  Each of these
           examples produces two distinct face-connected subdomains.}
  \figlab{threedcases}
\end{figure}
\begin{proof}
  We proceed by induction over $L$.
  Again, the case $L = 0$ leaves nothing to prove.  If the segment of the curve
  is contained in one level $L - 1$ subtree, the proof is finished by
  induction.  Else we have one subtree in which the segment begins, zero to
  $2^{d} - 2$ fully covered subtrees, and one subtree in which the segment
  ends.  To the first nonempty subtree we can apply Corollary~\ref{illcorlast},
  while Proposition~\ref{illpropfirst} applies to the last one.  Thus we know that
  the possibly incomplete subtrees lead to one connected piece each.  The case
  of two nonempty subtrees is thus completed and it remains to consider three
  or more.

  Now, whenever any two adjacent nonempty subtrees have even-odd numbers, they
  are face-connected since at least one of them must be full.  This covers the
  remaining three- and four-subtree cases in two dimensions.  In three
  dimensions, this clears all situations with three non-empty subtrees.  Since
  we can further reduce the number of remaining cases by symmetry, it remains
  to examine the subtree ranges $(i, \ldots, i + 3)$ through $(i, \ldots, 7)$
  for $i = 0, \ldots, 3$.  All of these cases satisfy our claim; we
  illustrate a few in \figref{threedcases}.
\end{proof}
We have completed the necessary proofs for a uniform space division into cubes
in $d \le 3$.
In \secref{illustrated-arbitrary} we extend the proof to arbitrary
dimension $d$.
The case of adaptive space divisions is considered in \secref{main}.

\subsection{The simplicial case}
\seclab{illustrated-simplicial}

In this subsection we examine the number of face-connected components of a
segment of the tetrahedral Morton SFC, $d = 2$ or $3$.
As we show in Figure~\ref{fig:face_conn_ex}, there exist cases where the number
of face-connected components in a uniform 2D level $L$ refinement can be as
high as $2(L-1)$.
We show that this is in fact a sharp upper bound.
We also show that in three dimensions the number of face-connected components
does not exceed $2L+1$. There exists an example with $2L$ face-connected
components and we conjecture that $2L$ is in fact the sharp estimate.
The proof of these bounds is fairly analogous to the results for cubes and 
relies and a divide-and-conquer approach by splitting the segment into 
subsegments of which we know the number of face-connected components.
The main difference to the cubical Morton curve is that we do not have a
strong symmetry property like \eqnref{Rdef}, and thus an analogue to
Corollary~\ref{illcorlast} only exists in a weaker form.

\begin{figure}
   \begin{minipage}{0.6\textwidth}
   \center
   \def\svgwidth{0.7\textwidth}   
   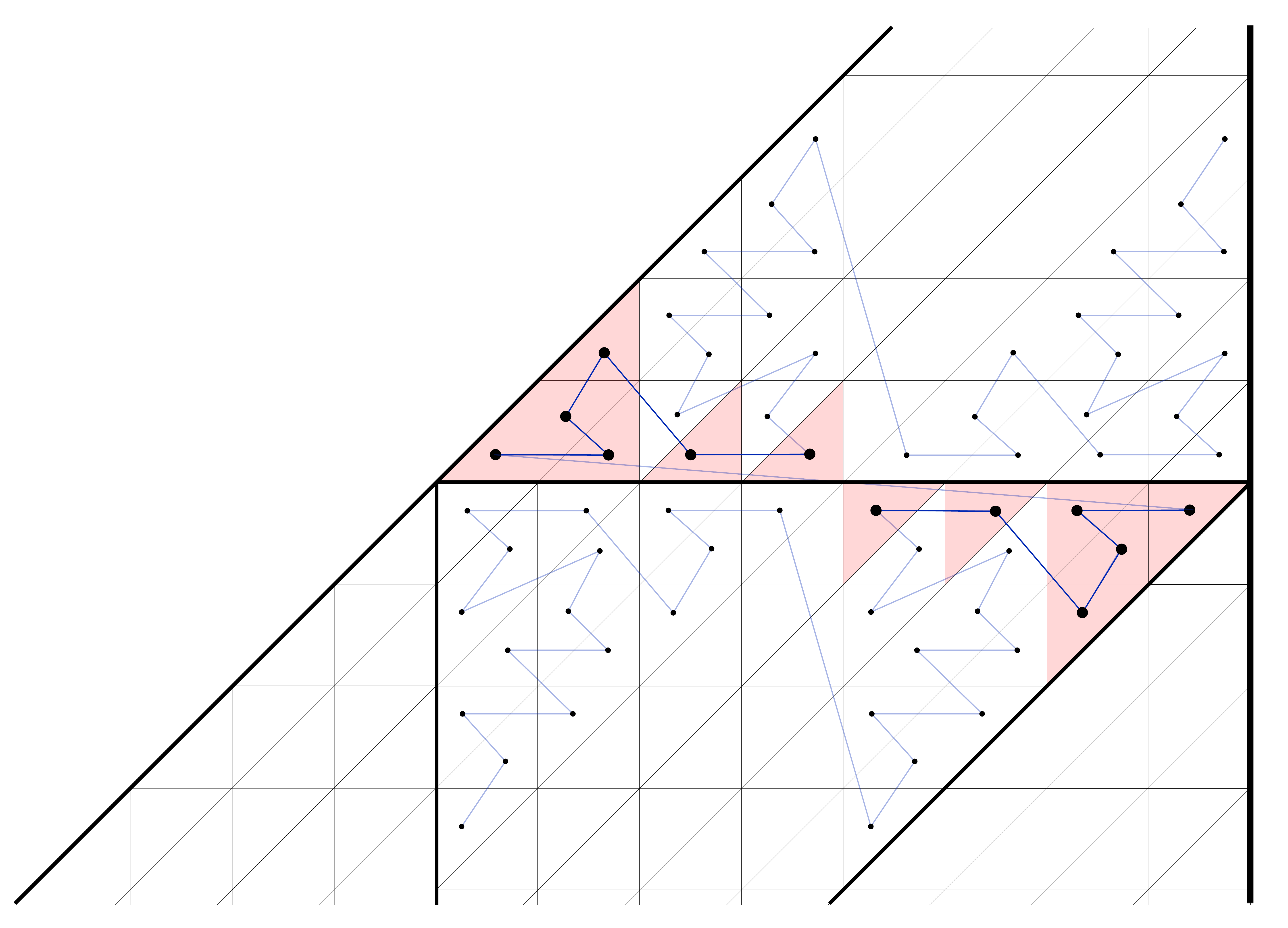
   \end{minipage}
   \begin{minipage}{0.38\textwidth}    
   \center
   \def\svgwidth{0.8\textwidth}   
   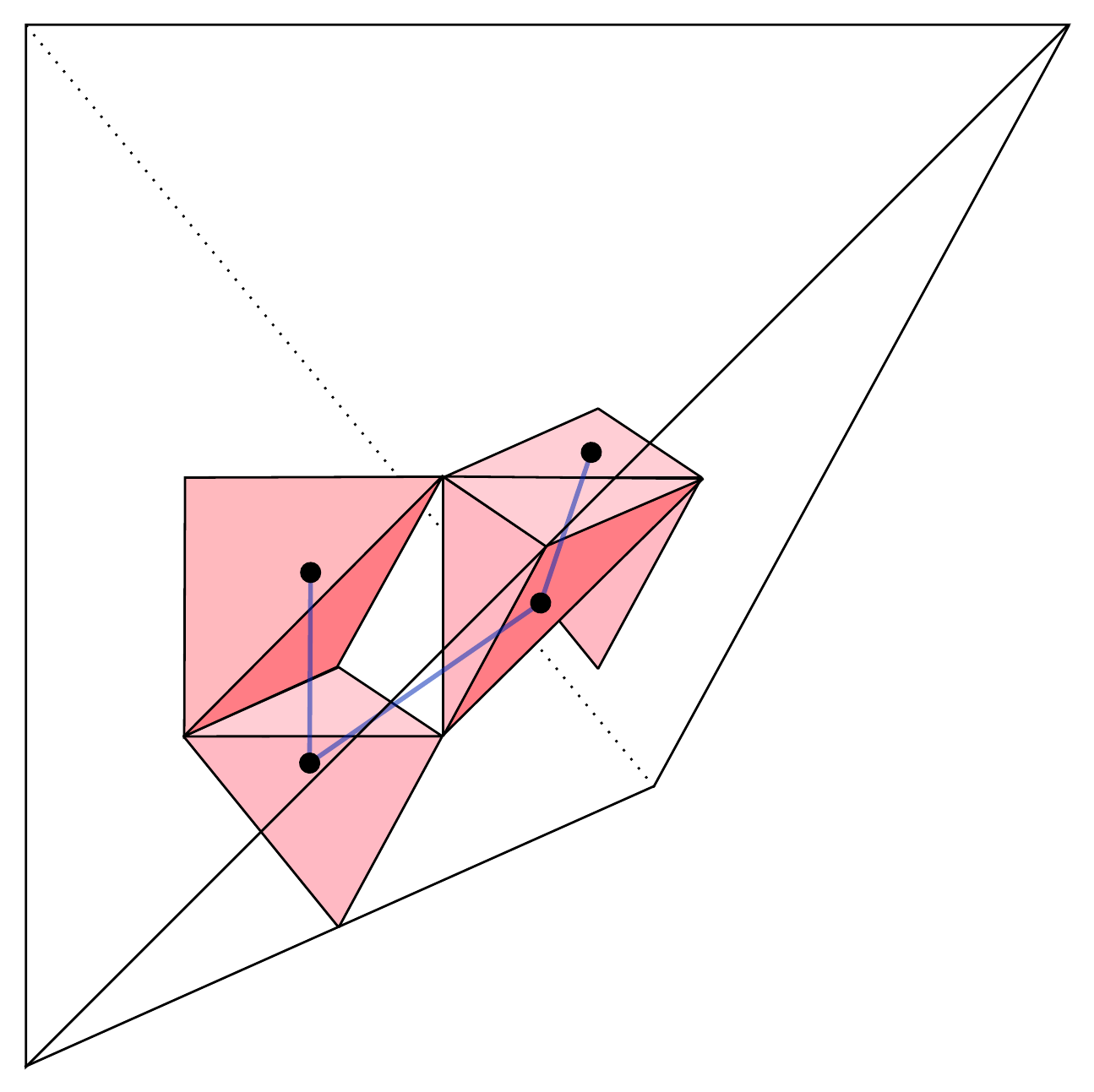
   \end{minipage}
   \caption{Left: a segment of the 2D SFC on a level $4$ refinement of $T^0_2$
            with six face-connected components (shaded pink).
   The number of face-connected components in 2D can be as high as $2(L-1)$;
   this estimate is sharp.
   Right: a 3D level 2 refinement of $T^0_3$ with four ($=2L$) face-connected
   components.
   We prove that an upper bound on the number of face-connected components is
   $2L+1$ and conjecture that $2L$ is sharp.
  }
  \label{fig:face_conn_ex}
\end{figure}%

\begin{lem}
  \label{lem:faceconn2dhelp}
  The following two properties hold for the TM-index in 2D, where we consider a
  uniform level $L$ refinement of an initial type 0 triangle $T$.
  \begin{itemize}
    \item Each type 1 subsimplex is face-connected to a type 0
      subsimplex with a greater TM-index.
    \item Each type 0 subsimplex that is also a descendant of the level 1,
      type 1 subtriangle $T_3$
      is face-connected to a type 1 subsimplex with a greater TM-index.
  \end{itemize}
\end{lem}
\begin{proof}
  The respective face-neighbor is the top face-neighbor for the type 1 subsimplex
  and the face-neighbor along the diagonal face for the type 0 subsimplex;
  see Figure~\ref{fig:faceconn2dhelp}.
  For type 0 we additionally require that the subsimplex is
  a descendant of $T_3$, since this ensures that the face-neighbor along the
  diagonal face is inside the root triangle.
  Despite this detail, the proofs for both items are identical, and we only
  present one for the first.

  Let $S$ denote an arbitrary type 1 subsimplex of level $L$ and let $S'$ be
  its neighbor across the top face.
  If $S$ and $S'$ share the same parent $P$ then there are two cases,
  which we also see in Figure~\ref{fig:haverkortSFC}:
  Either $\type(P) = 0$, then the local index of $S$ is 2 and that of $S'$ is
  3, or $\type(P) = 1$, in which case the local index of $S$ is 0 and that
  of $S'$ is 1.
  Thus, in both cases the TM-index of $S$ must be smaller than that of $S'$.
  We suppose now that $S$ and $S'$ have different parents, which implies
  $L\geq2$, and denote these different level $L-1$ subsimplices by $P$ and
  $P'$.
  The only possible combination is that $\type(P) = 1$ and $\type(P') = 0$, and
  that $P$ and $P'$
  are neighbors along $P$'s top face. Therefore, by an induction argument,
  $m(P) < m(P')$,
  and since the TM-index preserves the local order under refinement, each child
  of $P$ has a
  smaller TM-index than each child of $P'$.
  In particular we find $m(S)<m(S')$.
\end{proof}

\begin{figure}
  \center
  \def\svgwidth{0.8\textwidth}
  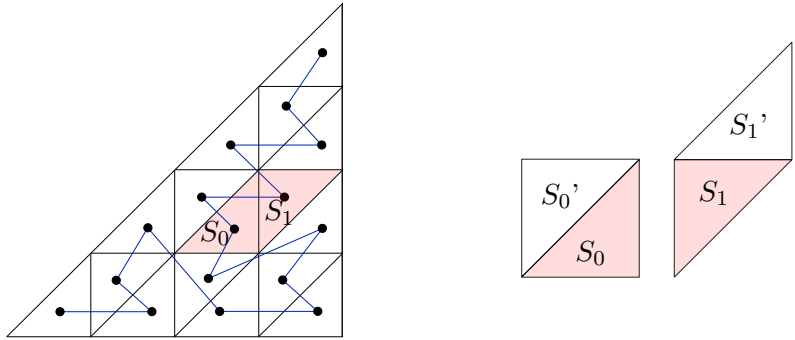
  \caption{Illustration of Lemma \ref{lem:faceconn2dhelp}.
   In 2D, choose any subsimplex $S_\ast$. If its neighbor along the top face
   $S_\ast'$ is inside the root triangle, then $m(S_\ast) < m(S_\ast')$.
   This condition is always fulfilled by any type 1 triangle and by type 0
   triangles that are descendants of the middle level 1 subtriangle.
  }
  \label{fig:faceconn2dhelp}
\end{figure}

Let us now show a 2D analogue to Proposition~\ref{illpropfirst}.

\begin{lem}
  \label{lem:faceconn2d}
 Consider a triangle $T$ that is uniformly refined to level $L$.
 If $T$ has type 0, then a contiguous segment of the SFC ending in the last
 level $L$ subsimplex has just one face-connected component.
 If $T$ has type 1, then this holds for segments starting in the first level
 $L$ subsimplex.
\end{lem}
\begin{proof}
  We only show the statement for $\type (T)=0$, since we can then use
  the symmetry of the 2D curve (Remark \ref{rem:simplexsym}) to obtain
  the result for the case $\type (T)=1$.
  We proceed by induction over $L$.

  For $L=0$ there is only one possible segment and it is connected.
  For $L>0$, let $j\in\{0,1,2,3\}$ be the local index of the level
  1 subtree $T'$ of $T$ in which the first level $L$ subsimplex of the segment lies.
  If $j\in\{0,1,3\}$, then the type of $T'$ is 0 and the statement follows
  by induction with the same argument as in the proof of
  Proposition~\ref{illpropfirst}.
  Thus, let $j=2$, i.e., the segment starts in the type 1 subtree of $T$.
  The part of the segment that is not inside $T'$ is the full last subtree of
  $T$ (local index 3) and thus it is face-connected.
  With Lemma \ref{lem:faceconn2dhelp} we conclude that each subsimplex in the
  subsegment in $T'$ is face-connected to a simplex with greater TM-index.
  Since this holds particularly for the last subsimplex in $T'$, the subsegment
  is also face-connected to the subsegment making up the last subtree.
  Thus, the whole segment is face-connected.
\end{proof}

For all other segments beginning with the first or ending in the last level $L$
subsimplex, and notably for all of those segments in 3D, we obtain an upper
bound of $L+1$ face-connected components, which we show in the next two lemmas.


\begin{lem}
\label{lem:faceconnlem0}
 Let a segment of the space-filling curve for a uniform level $L$ refined $d$-simplex
 consist of several full level 1 subsimplices plus
 one single
 level $L$ simplex at the end or at the beginning, then this segment has at
 most two face-connected components.
\end{lem}

\begin{proof}
  Similarly to the last paragraph in the proof of Proposition \ref{illpropboth},
  and in analogy to \figref{threedcases}, we can show this claim by enumerating
  all possible cases.
\end{proof}
\begin{lem} 
\label{lem:faceconnlem1}
 If a $d$-simplex is uniformly refined to level $L$, then any
 segment of the space-filling curve ending in the last subsimplex or starting in
 the first has at most $L+1$ face-connected components.
\end{lem}
\begin{proof}
  Consider the case that the segment starts in the first simplex.
  For $L=0$ there is only one possible segment consisting of the unique level
  0 subsimplex and it is thus connected.
  Let now $L>0$.
  Since the segment begins at the very first level $L$ subsimplex, we can
  separate it into two parts.
  The first part at the beginning consists of 0 to $2^d-1$ full level 1 subtrees,
  and the second part is one possibly incomplete level 1 subtree.

  By the induction assumption, the second part has at most $L$ face-connected
  components.
  From Lemma \ref{lem:faceconnlem0} we obtain that the first part
  together with the first level $L$ subsimplex of the second part has at most
  two face-connected components.
  Since this first level $L$ subsimplex is contained in one of the components
  of the second part, we obtain
  \begin{equation}
    L + 2 - 1 = L + 1
  \end{equation}
  components in total.

  If the segments ends in the last simplex, the order of parts is reversed.
  The first part of the segment is the part in the level 1 subtree where
  the segment starts, and the second part consists of the remaining full
  level 1 subtrees.
  We obtain the bound on the number of face-connected components using the same
  inductive reasoning as above.
%
\end{proof}

We have so far argued the connectivity of specific kinds of SFC segments.
This suffices to proceed to arbitrary segments of the tetrahedral Morton SFC.
\begin{prop}
\label{prop:faceconncomp}
 Any contiguous segment of the space-filling curve of a uniform level $L\geq2$
 refinement of a type 0 simplex has at most $2(L-1)$ face-connected components in 2D and
 $2L+1$ face-connected components in 3D.
 For $L=1$, there are at most two face-connected components
 (this applies to both 2D and 3D).
\end{prop}
\begin{proof}
 We first show that for $d \le 3$ the number of face-connected
 com\-po\-nents is bounded by $2L+1$:
 If a given segment is contained in a level 1 subtree, we are done by induction.
 Otherwise we can divide the segment into three (possibly empty) pieces:
 First,
 the segment in one incomplete level 1 subtree ending at its last level $L$
 subsimplex,
 then one contiguous segment of full level 1 subtrees
 and finally a segment in one (possibly incomplete) level 1 subtree that
 starts at its first level $L$ subsimplex.
 Lemma~\ref{lem:faceconnlem1}
 implies that the first and the last piece have at most $L$ face-connected
 components each.
 By Lemma~\ref{lem:faceconnlem0}, the second piece has one or two
 face-connected components, and if the number is two, then it is face-connected
 to the first or to the third piece.
 Thus, it adds only one face-connected component to the total number, and we
 obtain at most
 \begin{equation}
  L + 1 + L = 2L + 1
 \end{equation}
 face-connected components.

Let us now specialize to 2D.
We conclude from Lemma \ref{lem:faceconn2d} that
the first subsegment only adds more than one face-connected component if
it is contained in the only level 1 subtree of type 1 (local index 2).
Similarly, the second subsegment only adds more than one face-connected component
if it is contained in a level 1 subtree of type 0.
In particular, if both subsegments add more than one connected component, the
second subsegment is contained in the last level 1 subtree (local index 2).
Thus, the middle subsegment is empty in this case.

If both of these subsegments have less than $L$ face-connected components, there
is nothing left to show since the overall number of components is then less
than or equal to $2(L-1)$.
So suppose that one of the subsegments has $L$ face-connected components and the other one
has at least $L-1$.
We depict this situation in Figure~\ref{fig:explaintrianglprop}.
We observe that the first and second level $L$ simplex in this first segment
are face-connected to the first and second level $L$ simplex in the second segment.
If, however, the second subsegment has $L$ connected components then its last two 
level $L$ simplices are face-connected to the last two level $L$ simplices of the first 
subsegment.

We thus can subtract two connected components from the total count, which leads
to at most
\begin{equation}
  L + L - 2 = 2(L-1)
\end{equation}
face-connected components in total.
%
%
\end{proof}

\begin{figure}
  \center
  \includegraphics[width=0.6\textwidth]{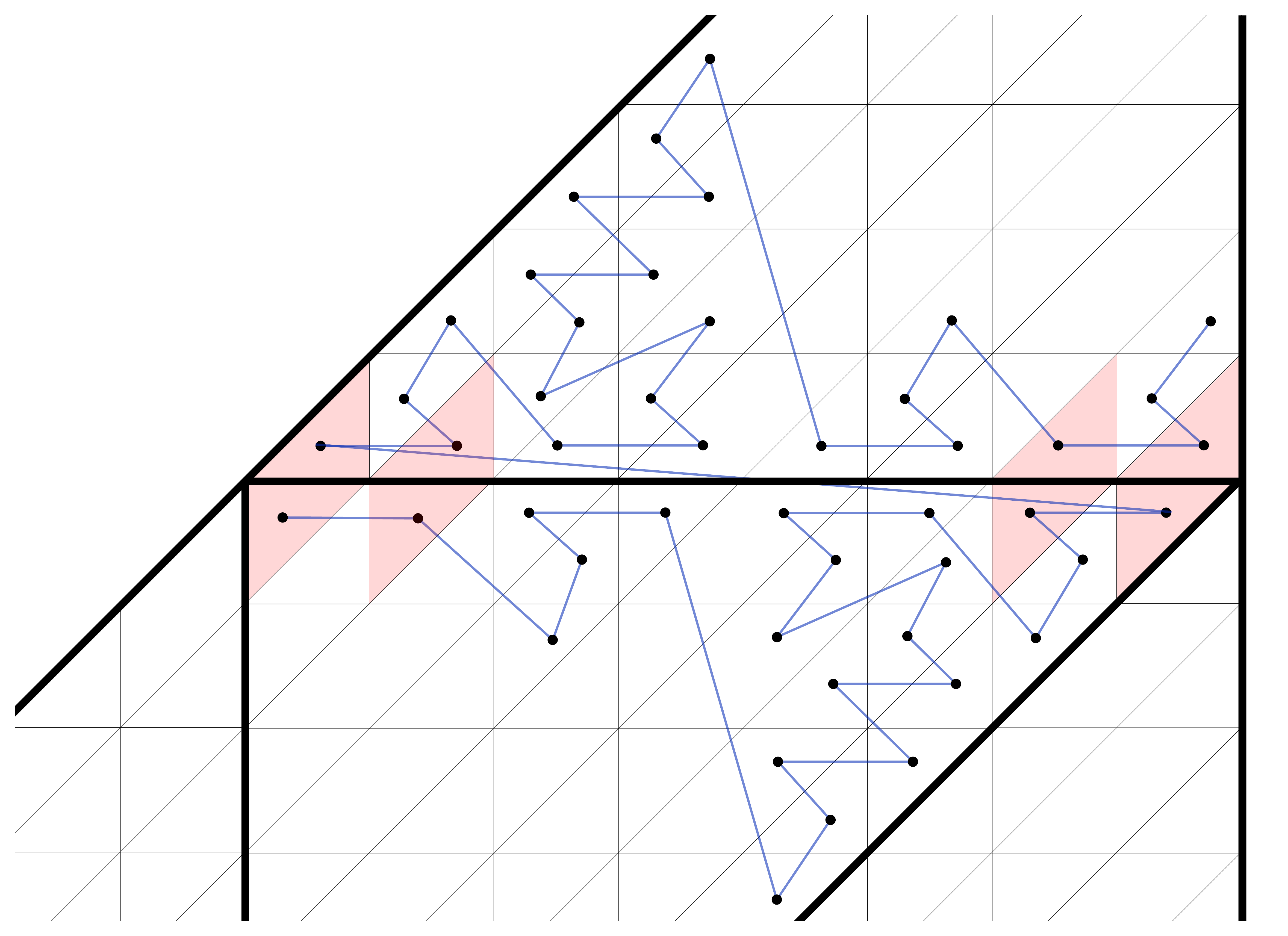}
  \caption{An illustration of the 2D case in the proof of Proposition \ref{prop:faceconncomp} for 
  $L=4$. The bottom segment has the maximal number of $L$ face-connected components. 
  Since its first and second triangle (on the left, shaded in pink) are connected with the
  top segment, the possible number of connected components is reduced by two.
  If however, the second segment has $L$ face-connected components then its last two
  triangle (on the right) are connected with the bottom segment.
  Thus, the number of face-connected components is less than or equal to $2L-2$.}
  \label{fig:explaintrianglprop}
\end{figure}

We briefly discuss whether we can sharpen these bounds.
In 2D, this is not possible by counterexample; see Figure~\ref{fig:face_conn_ex}.
In 3D, we construct a segment with $2L$ face-connected components using the
SFC-indices 22--25 of a uniform level 2 refinement of a type 0
tetrahedron.
We believe that the case that the first and the last piece described in
the proof of Proposition \ref{prop:faceconncomp} have $L$ face-connected
components each and that additionally the middle piece adds one component
does not occur.
\begin{conj}
\label{con:3TMconjecture}
 In 3D, the number of face-connected components is bounded by $2L$.
 This estimate is sharp.
\end{conj}

\section{Proofs for arbitrary dimension}
\seclab{arbitrary}

The construction of the cubical Morton curve generalizes readily to
arbitrary space dimension $d$.
We should suppose that the main result (that any contiguous segment consists of
at most two face-connected subdomains) generalizes as well.
Indeed, we propose two different ways to prove this in the following
Sections~\ref{sec:illustrated-arbitrary} and \ref{sec:nonrecursiveproofs}.
The first is closer to the geometric approach we have been using in
\secref{illustrated}, while the second is more formal and paves the way for
quantitative studies of the frequency of disconnections in \secref{enumeration}
below.

We close this section with the extension of the proofs for both cubical and
tetrahedral curves from uniform to adaptive meshes (see \secref{main}), which
is the remaining step to establish Theorems~\ref{illthmallcube} and
\ref{illthmalltets},
and discuss implications for a forest of octrees (\secref{forest}).

\subsection{Induction proofs for $d$-cubes}
\seclab{illustrated-arbitrary}

We use induction over both the dimension and the level of subdivision to
prove the main statement for all dimensions $d > 0$.
These proofs imply the statements of the previous \secref{illustrated-cubical} as
special cases.  For convenience we denote any $d$-tant as a quadrant.  We make
use of the following definition of subtree ranges.
\begin{dfn}
  \label{defintervals}
  Let the space dimension be $d > 0$.  For any $0 \le d' \le d$ and $0 \le k <
  2^{d-d'}$, we define the following interval containing $2^{d'}$ integers,
  \begin{equation}
    \eqnlab{Idk}
    \Idk = 2^{d'} [ k, k + 1 ) .
  \end{equation}
  We use this interval to denote a specific contiguous range of subtree
  indices.
\end{dfn}
We define the following auxiliary statements, first considering a one-sided
segment and then a general two-sided one.
\begin{prop}
  \label{arboneended}
  If a segment of a Morton curve is fully contained in the level 1 subtrees
  enumerated by a given $\Idk$ and contains the first or last subquadrant in
  this range of subtrees, then it corresponds to one face-connected subvolume.
\end{prop}
\begin{proof}
  By symmetry of the Morton curve, we can restrict the discussion to the case
  of the first subquadrant.  Let us begin by proving the statement for
  subdivision level $L = 1$.  By \eqnref{Idk} the lowest subtree index in the
  segment is $k2^{d'}$.  This number has $d'$ zero bits from the right.  All
  other indices in $\Idk$ have one or more ones in the lower $d'$ bits while
  being bitwise identical in the higher bits.  For any of these indices we can
  flip the low bits to zero one by one, effectively transitioning through face
  neighbors and monotonously decreasing the index until we reach $k2^{d'}$.
  This whole sequence of face-connected subtrees is contained in $\Idk$.  In
  conclusion, all trees in $\Idk$ are face-connected to $k2^{d'}$ and thus to
  each other.

  Now let $L > 1$ and assume the above statement for $L - 1$.  To prove it for
  $L$ we make an induction over $d'$.  If $d' = 0$ we have a single subtree and
  can readily invoke the induction assumption for $L - 1$.  Else there are two
  possible cases: Either the segment is fully contained in one of $\Idmk$ or
  $\Idmkp$ and we apply the induction over $d'$.  Otherwise $\Idmk$ contains
  full subtrees only and the segment reaches into $\Idmkp$.  Each nonempty
  subtree $j$ in the latter interval must contain its first subquadrant, which
  has a face connection to the full tree $j - 2^{d'-1} \in \Idmk$.  Since by
  the proof for $L = 1$ all subtrees in $\Idmk$ are face-connected, we are
  done.
\end{proof}
\begin{prop}
  \label{arbtwoended}
  If a segment of a Morton curve is contained in the level 1 subtrees $\Idk$,
  it produces no more than two distinct face-connected subvolumes.
\end{prop}
\begin{proof}
  Again let us prove the statement first for $L = 1$.  If $d' = 0$ we have just
  one level 1 subquadrant that clearly satisfies our claim.  For positive
  $d'$ we distinguish the following cases.  If the segment is fully contained
  in either $\Idmk$ or $\Idmkp$, we apply the induction on $d'$.  Else we know
  that the last subquadrant of $\Idmk$ and the first of $\Idmkp$ are in the
  segment.  By Proposition~\ref{arboneended} we have at most two disconnected
  pieces and the statement holds.
  
  If $L > 1$ the case $d' = 0$ reduces to the same statement for $L - 1$ and we
  are done by applying the induction over $L$.  Else, the proof proceeds
  unchanged as above with the desired result.
\end{proof}
We have implicitly proved the main result for any uniform level $L$
subdivision, since a level 0 subtree trivially satisfies our claim, and
otherwise the root cube is the union of the level 1 subtrees $I^d_0$.

\subsection{A non-inductive proof for $d$-cubes}
\seclab{nonrecursiveproofs}


In this section we elaborate on the formalism of the Morton index
(see \secref{concepts-cubical}) to obtain the result without induction.
The tool we use is the map $\Omega$ from the index
$Q=(q^1\dots q^L)_2$ to a subset of $\mathbb{R}^d$ stated in \eqnref{rdset}.
For $1\leq r \leq d$, we define the coordinate along axis $r$,
\begin{equation}
  \eqnlab{Qaxisr}
  Q_r = (q^1_r\dots q^L_r)_2.
\end{equation}
The map $\Omega(Q)$ may be written as
\begin{equation}
  \eqnlab{Qmapsimpler}
  \Omega(Q) = [2^{-L}Q_1, 2^{-L}(Q_1 + 1)] \times \dots \times
  [2^{-L}Q_d, 2^{-L}(Q_d + 1)].
\end{equation}

\begin{lem}
 If $Q$ and $\tilde{Q}$ are such that $\tilde{Q}_k \leq Q_k$ for all $1\leq
 k\leq d$, then $\tilde{Q} \leq Q$.
\end{lem}
\begin{proof}
  The order of the bits in $Q_k$ ($\tilde{Q}_k$) is the same as their order in
  $Q$ ($\tilde{Q}$).
\end{proof}

\begin{thm}
  \thmlab{Y0}
  For any index $Q^\text{end}$, the interior of $Y_0=\cup_{Q=0}^{Q^{\text{end}}}
  \Omega(Q)$ is star-shaped (and thus $Y_0$ is face-connected and contractible).
\end{thm}
\begin{proof}
  If $Q\in \{0,\dots,Q^{\text{end}}\}$, then so are all $\tilde{Q}$ such that
  $\tilde{Q}_k\leq Q_k$ for all $1\leq k \leq d$.  The domains of these
  quadrants define a box between the origin and the corner of $\Omega(Q)$
  farthest from the origin,
  \begin{equation}
    B(Q) := [0,2^{-L} (Q_1 + 1)]\times \dots \times [0,2^{-L} (Q_d + 1)].
  \end{equation}
  Indeed, $Y_0$ is the union of these boxes, $Y_0 = \cup_{Q=0}^{Q^{\text{end}}}
  B(Q)$, and the union of their interiors is the interior of $Y_0$.  Each of
  these boxes contains the midpoint of $\Omega(0)$ in its interior and is
  star-shaped with respect to it.  Therefore the interior of $Y_0$ is
  star-shaped with respect to that point as well.
\end{proof}

\begin{cor}
  \label{cor:start}
  If $Q \vert Q^{\text{start}} = Q$ ($\vert$ means bitwise-or) for all
  $Q\in\{Q^{\text{start}},\dots, Q^{\text{end}}\}$, then the interior of
  $Y=\cup_{Q=Q^{\text{start}}}^{Q^{\text{end}}} \Omega(Q)$ is star-shaped.
\end{cor}
\begin{proof}
  If $Q\vert Q^{\text{start}} = Q$, then the 1-bits of $Q-Q^{\text{start}}$
  are a subset of the 1-bits of $Q$.  One can then verify that
  $(Q-Q^{\text{start}})_k = Q_k - Q_k^{\text{start}}$ for all $1\leq k\leq d$.
  Therefore $\cup_{Q=Q^{\text{start}}}^{Q^{\text{end}}}
  \Omega(Q-Q^{\text{start}})$ is $Y$ translated by the vector
  $(-2^{-L}Q_1^{\text{start}},\dots,-2^{-L}Q_d^{\text{start}})$.  This is the
  same as $\cup_{Q=0}^{Q^{\text{end}} - Q^{\text{start}}} \Omega(Q)$,
  which is star-shaped by \thmref{Y0}.
\end{proof}

\begin{cor}
  \label{cor:end}
  If $Q \& Q^{\text{end}} = Q$ ($\&$ denotes the bitwise and operator) for all
  $Q\in\{Q^{\text{start}},\dots, Q^{\text{end}}\}$, then the interior of 
  $Y$ is star-shaped.
\end{cor}
\begin{proof}
  Mirroring every quadrant about the midpoint of the unit cube does not change
  the shape of $Y$.  The mirror of $\Omega(Q)$ is $\Omega(R(Q))$, where $R(Q)$
  denotes the bitwise negation \eqnref{Rdef}.
  Therefore $R(Q)\vert R(Q^{\text{end}}) =
  R(Q \& Q^{\text{end}}) = R(Q)$ for all $R(Q) \in \{R(Q^{\text{end}}),\dots,
  R(Q^{\text{start}})\}$.
\end{proof}

\begin{thm}
  \label{thm:uniontwo}
  The interior of
  $Y$ is
  star-shaped, or $Y$ is the union of two sets whose interiors are
  star-shaped.
\end{thm}
\begin{proof}
  Let $\tilde{q}$ be the most significant bits common to all of
  $\{Q^{\text{start}},\dots,Q^{\text{end}}\}$.  We can split the segment into
  $\{Q^{\text{start}},\dots,(\tilde{q}011\dots1)_2\}$ and
  $\{(\tilde{q}100\dots0)_2,\dots,Q^{\text{end}}\}$.  The interior of the
  domain of the first segment is star-shaped by Corollary~\ref{cor:end}; the
  interior of the domain of the second segment is star-shaped by
  Corollary~\ref{cor:start}.
\end{proof}

\subsection{From uniform to adaptive meshes}
\seclab{main}

We have completed the necessary proofs for a uniform space division, in the
case of cubical refinement for any space dimension $d$, and previously for
triangular and tetrahedral refinement (see \secref{illustrated-simplicial}).
As we state in this section, an adaptive space division does not require any
more effort (see also \cite[page 176]{Bader12}).
\begin{proof}[Proof of Theorems~\ref{illthmallcube} and \ref{illthmalltets}]%
  Any adaptive tree of quadrants with level $\le L$ can be refined into level
  $L$ quadrants exclusively.  This operation does not change the connectivity
  between boundaries of the designated subdomain.  In particular, the number of
  face-connected subdomains remains unchanged and the proof reduces to applying
  Propositions~\ref{illpropboth}, \ref{prop:faceconncomp} (only $d \le 3$)
  or~\ref{arbtwoended} (any $d$) above.
\end{proof}

\subsection{From one tree to a forest}
\seclab{forest}


If we consider a forest of octrees as in \cite{StewartEdwards04,
BangerthHartmannKanschat07, BursteddeWilcoxGhattas11}, a contiguous segment of
the Morton curve may traverse more than one tree.
In this case, the segment necessarily contains the last subquadrant of any
predecessor tree, as well as the first subquadrant of any successor tree in the
segment.
In the cubical case, we know by Proposition~\ref{arboneended} that no jumps can
occur at all (when not counting the transition between two successive trees as
a jump).
For the simplicial case, we may use Lemmas~\ref{lem:faceconn2d} and
\ref{lem:faceconnlem1} to obtain
the bounds $L+1$ (2D) and $2L+1$ (3D).

\section{Enumeration of face-connected segments}
\seclab{enumeration}

We would like to examine not only how many pieces an SFC segment can have, but
also how frequently segments of different numbers of pieces occur.
To this end, we propose a theoretical lower bound for the cubical case
and supply numerical studies for both cubical and tetrahedral SFCs.

\subsection{Lower bound on fraction of continuous segments}
\seclab{prooflowerbound}


\begin{thm}
  \label{thm:lowerbound}
  The fraction of continuous segments of length $l$ of the level-$L$
  $d$-dimensional, cubical Morton curve is
  \begin{equation}
    \phi_{d,L,l} \geq \frac{1}{2^d  -1}.
  \end{equation}
\end{thm}
\begin{proof}
  Let $s=\{Q^{\text{start}},\dots,Q^{\text{end}} = Q^{\text{start}}+l-1\}$ be
  the first discontinuous segment of length $l$.  As in the proof of
  Theorem~\ref{thm:uniontwo}, $s$ divides into two continuous segments,
  $\{Q^{\text{start}},\dots,(\tilde{q}011\dots1)_2\}$ and
  $\{(\tilde{q}100\dots0)_2,\dots,Q^{\text{end}}\}$, where $\tilde{q}$ are
  significant bits that are common to all numbers in the segment.  By the same
  reasoning as used in Corollary~\ref{cor:start}, the shapes of the two segments
  are not affected by $\tilde{q}$: changing $\tilde{q}$ translates the whole
  domain.  Therefore, as this is the first discontinuous segment, $\tilde{q}$
  must be $(0\dots 0)$, i.e., the two continuous pieces of the segment are
  $s^-=\{Q^{\text{start}},\dots,2^k-1\}$ and $s^+=\{2^k,\dots,Q^{\text{end}}\}$
  for some $k$.  Since $Q^{\text{end}}$ cannot have a more significant 1-bit
  than $2^k$, $Q^{\text{end}} \leq 2^{k+1} - 1$.  We note that bitwise negation
  of the first $k+1$ bits, $\tilde{R}(Q):=2^{k+1}-1-Q$, induces a map
  $\Omega(Q)\mapsto \Omega(\tilde{R}(Q))$ that is the reflection about the
  midpoint of the box formed by $\cup_{Q=0}^{2^{k+1}-1} \Omega(Q)$.

  We first want to find a lower bound for $Q^{\text{start}}$.  We note that
  $2^k\in s^+$, so if $\Omega(Q^*)$ is a face neighbor of $\Omega(2^k)$, then
  $Q^*\not\in s^-$, and if $Q^* < 2^k$, then because of the definition of $s^-$,
  $Q^*$ is less than every $Q\in s^-$, including $Q^{\text{start}}$. 
  Let $j=k \spmod d$ and let $j'=d - j - 1$.
  The $Q^*$ that plays the crucial role is the neighboring quadrant that is
  closer to the origin in
  the $(j+1)$th direction.  We have to subtract one from the $(j+1)$th
  coordinate of $2^k$: i.e., $(2^k)_{j+1}=2^{\lfloor k / d \rfloor}$ in the
  notation of \eqnref{Qaxisr}, so
  $Q^*_{j+1} =2^{\lfloor k / d \rfloor} - 1$, while the other coordinates are
  the same, viz.\ zero.  Therefore
  \begin{equation}
    \begin{aligned}
      Q^*
      &= (0\dots0\overbrace{%
        \underbrace{0\dots0}_{\text{$j'$-times}}1%
        \underbrace{0\dots0}_{\text{$j$-times}}%
      }^{\text{$\lfloor k / d\rfloor$-times}})_2
      & &= \sum_{i=1}^{\lfloor k / d \rfloor} 2^{(i-1)d + j} \\
      &= \sum_{i=1}^{\lfloor k / d \rfloor} 2^{(\lfloor  k / d \rfloor - i)d +
      j}
      & &= 2^k\sum_{i=1}^{\lfloor k / d \rfloor} 2^{-i d} \\
      &= 2^k\left(\sum_{i=1}^{\infty} 2^{-i d} - \sum_{i=\lfloor k /d \rfloor
        + 1}^{\infty} 2^{-id}\right)
      & &= 2^k\frac{1}{2^{d}-1} - 2^k
      \sum_{i=\lfloor k /d \rfloor + 1}^{\infty} 2^{-id} \\
      &= 2^k\frac{1}{2^{d}-1} - 2^j\sum_{i=1}^{\infty}
      2^{-id}
      & &= \frac{2^k - 2^j}{2^d-1}.
    \end{aligned}
  \end{equation}

  By the definition of $Q^\text{start}$, there is a continuous segment of
  length $l$ that starts at each $Q\in \{0,\dots, Q^*\}$.  For each of these
  continuous segments, there is another continuous segment, obtained by the
  reflection map $\tilde{R}$, than ends with
  $Q\in\{\tilde{R}(Q^*),\dots,2^{k+1}-1\}$.  We want to show that these two
  sets of continuous segments are distinct, i.e., that there is no segment of
  length $l$ that starts with $Q\leq Q^*$ and ends with $Q+l-1\geq
  \tilde{R}(Q^*)$.  We thus have to show that the shortest segment with
  endpoints in each set, $\{Q^*,\dots,\tilde{R}(Q^*)\}$ is longer than $l$,
  i.e., $l< \tilde{R}(Q^*) - Q^* + 1 = 2^{k+1} - 2Q^*$.

  We will prove this bound by finding an upper bound for $Q^{\text{end}}$.
  We note that $\tilde{R}(2^k)=2^k-1\in s^-$, so $\Omega(\tilde{R}(Q^*))$ is a face
  neighbor of $\Omega(2^k-1)$, so by the same reasoning as above,
  $\tilde{R}(Q^*)$ must be greater than $Q^{\text{end}}$, and thus
  \begin{equation}\eqnlab{lbound}
    l = Q^{\text{end}} - Q^{\text{start}} + 1
    \leq (\tilde{R}(Q^*) - 1) - (Q^* + 1) + 1
    =
    2^{k+1} - 2Q^* - 2.
  \end{equation}

  We have shown that there are at least $2(Q^*+1)$ continuous segments in the
  first $2^{k+1}$ segments, each of which begins and ends in the range
  $\{0,\dots,2^{k+1}-1\}$.  None of the numbers in these segments has more than
  $k+1$ significant bits, so by the same reasoning as in
  Corollary~\ref{cor:start}, adding a multiple of $2^{k+1}$ to each number in
  one of these segments is a translation of its domain, and is thus continuous.
  Therefore there are at least $2(Q^*+1)$ continuous segments for \emph{every}
  $2^{k+1}$ segments, and thus
  \begin{equation}
    \begin{aligned}
      \phi_{d,L,l} &\geq \frac{2(Q^*+1)}{2^{k+1}} = \frac{Q^*+1}{2^k}
      \\
      &=
      \frac{2^k - 2^j}{2^k(2^d-1)} + \frac{1}{2^k}
      \\
      &=
      \frac{2^k - 2^j + 2^d - 1}{2^k(2^d-1)}
      \\
      &\geq
      \frac{2^k}{2^k(2^d-1)}
      & \text{($j < d$ by definition)}
      \\
      &=
      \frac{1}{2^d - 1}.
    \end{aligned}
  \end{equation}%
\end{proof}

\subsection{Computational studies---cubical Morton curve}
\seclab{numerical}

Having shown that a segment of a Morton curve is composed of one or two
face-connected subdomains, a natural question to ask is how many of each type
there are.  More formally, we ask: for a given dimension $d$, recursive level
$L$, and segment length $l$, what fraction $\phi_{d,L,l}$ of the
$2^{dL}-(l-1)$ possible segments are in one face-connected piece? 

This question can be answered recursively.  Each segment of length $l>1$ on
level $L$ can refine to $2^{2d}$ segments on level $(L+1)$ with lengths between
$2^d (l-2) + 2$ and $2^dl$, as illustrated in \figref{segmentrecurse}.  We
divide connected segments into two categories: weakly connected, when the
first and last quadrants in the segment are (face-)adjacent, and strongly connected,
when they are not.  Disconnected segments only refine to disconnected
segments.  Strongly connected segments only refine to strongly connected
segments.  Weakly connected segments refine to all three types: how many of
each depends on the direction in which the first and last quadrants are
adjacent.

\begin{figure}\centering
  \begin{minipage}{0.32\textwidth}\centering
    \inputtikz{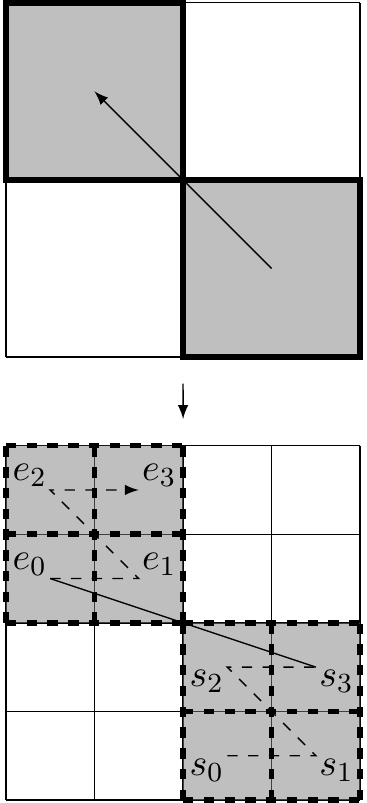}
  \end{minipage}
  \begin{minipage}{0.32\textwidth}\centering
    \inputtikz{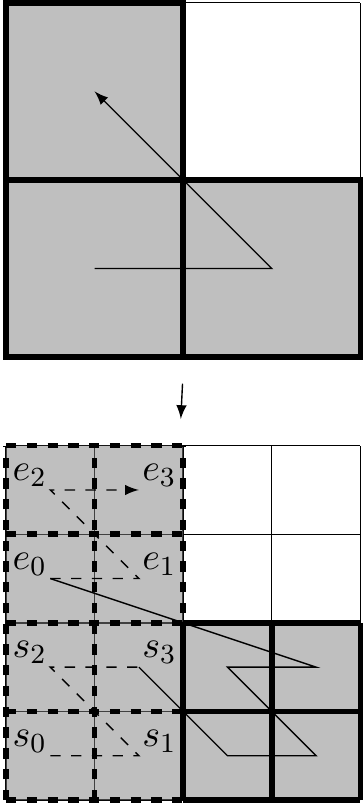}
  \end{minipage}
  \begin{minipage}{0.32\textwidth}\centering
    \inputtikz{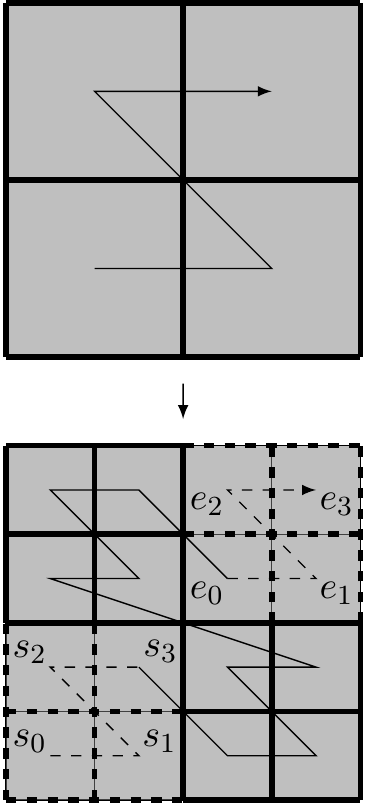}
  \end{minipage}
  \caption{%
    The refinement of disconnected (left), weakly connected (middle), and
    strongly connected segments (right).  Each coarse segment (top) refines to
    one of $2^{2d}=16$ possible refined segments (bottom), each starting with
    $s\in \{s_0,\dots,s_3\}$ and ending with $e\in \{e_0,\dots,e_3\}$.
    Disconnected segments refine to disconnected segments; strongly connected
    segments refine to strongly connected segments; weakly connected segments
    refine to disconnected (e.g., $\{s_3,\dots,e_0\}$), weakly connected
    (e.g., $\{s_2,\dots,e_0\}$), and strongly connected segments (e.g.,
    $\{s_0,\dots,e_3\}$).%
  }%
  \figlab{segmentrecurse}
\end{figure}%
We give pseudocode for this recursive calculation in the function Enumerate
(\algref{enumerate}).  This algorithm is implemented in the Python script
\texttt{morton.py}.%
\footnote{\url{https://github.com/cburstedde/p4est/tree/develop/doc/morton/morton.py}}

\begin{algorithm}
  \caption{%
    Enumerate ($d$, $L$, $l$)
  }%
  \alglab{enumerate}
  \DontPrintSemicolon

  \KwData{dimension $d\geq 1$, level $L\geq 0$, segment length $1\leq l\leq
  2^{dL}$.}

  \KwResult{$(n_d,n_s,n_{w,1},\dots,n_{w,d})$, the number of segments of
  length $l$ that are disconnected, strongly connected, and weakly connected
  in each direction.}

  \If%
  (\hfill[define segments of length 1 to be strongly connected])%
  {$l=1$ {\rm{\bf or}} $L=0$}%
  {%
    \KwRet{$(0,2^{dL},0,\dots,0)$}
    \Comment*{one segment for each quadrant}
  }%

  $(n_d,n_s,n_{w,1},\dots,n_{w,d}) \leftarrow (0,\dots,0)$

  $c\leftarrow \lceil l/2^d \rceil$
  \Comment*{compute the shortest length that can refine to length $l$}

  \lIf%
  {$l \spmod 2^d < 2$}%
  {%
    $C \leftarrow 1 + \lfloor l/2^d \rfloor$
  }%
  \lElse{%
    $C \leftarrow 2 + \lfloor l/2^d \rfloor$
  }%
  \Comment*{\ldots longest \ldots}

  \For{$k\in\{c,\dots,C\}$}{%

    \eIf{$k=1$}{%
      $N \leftarrow 2^{d(L-1)}$ \Comment*{\# of coarse quadrants}
      $(n_d,n_s,n_{w,1},\dots,n_{w,d}) \leftarrow
      (n_d,n_s,n_{w,1},\dots,n_{w,d}) +N*$ RefineOne ($d$,$l$)\;
    }%
    {%
      $r\leftarrow l - (k-2)2^d$ \Comment*{\# of children in two coarse
      end quadrants}
      $m\leftarrow \min\{(r-1),2^{d+1}-(r-1)\}$\Comment*{\# of ways to split
      between two families}

      $(N_d,N_s,N_{w,1},\dots,N_{w,d}) \leftarrow$ Enumerate ($d$, $L-1$, $k$)

      $n_d \leftarrow n_d + N_d * m$ \Comment*{disconnected $\rightarrow$
      disconnected}
      $n_s \leftarrow n_s + N_s * m$ \Comment*{strongly connected $\rightarrow$
      strongly connected}

      \For{$1\leq j\leq d$}{%
        $(n_d,n_s,n_{w,j}) \leftarrow (n_d,n_s,n_{w,j}) +N_{w,j}*$ RefineWeak
        ($d$, $j$, $r$)\;
      }%
    }%

  }%

  \KwRet{$(n_d,n_s,n_{w,1},\dots,n_{w,d})$}

\end{algorithm}

Enumerate calls on some lookup tables: RefineOne($d$,$l$) (\figref{refineone})
counts how many disconnected, strongly and weakly connected segments of length
$l$ are refined from one $d$-dimensional quadrant (the weakly connected
segments are broken down by the direction in which the first and last quadrant
are adjacent); RefineWeak($d$,$j$,$r$) (\figref{refineweak}) counts how many
disconnected, strongly and weakly connected segments with $r$ quadrants in the
end-families are refined from one weakly connected segment in direction $j$ (a
weakly connected segment only produces weakly connected segments in the same
direction). 

\begin{figure}\centering
  \inputtikz{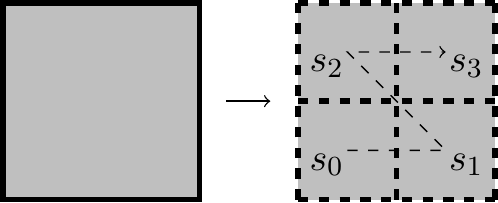}

  \null

  \renewcommand{\arraystretch}{1.1}
  \begin{tabular}{|c|l|}\hline
    RefineOne(2,1) & $n_s = 4$ ($\{s_0\}$, $\{s_1\}$, $\{s_2\}$,
    $\{s_3\}$)
    \\ \hline
    RefineOne(2,2) & $n_d = 1$ ($\{s_1,s_2\}$), $n_{w,1} = 2$ ($\{s_0,s_1\}$,
    $\{s_2,s_3\}$)
    \\ \hline
    RefineOne(2,3) & $n_{w,2} = 2$ ($\{s_0,s_1,s_2\}$, $\{s_1,s_2,s_3\}$)
    \\ \hline
    RefineOne(2,4) & $n_s = 1$ ($\{s_0,s_1,s_2,s_3\}$)
    \\ \hline
  \end{tabular}
  \caption{%
    We list the RefineOne($d$,$l$) tables used in Enumerate
    (\algref{enumerate}) for $d=2$ as an example.  Unlisted values are zero.%
  }%
  \figlab{refineone}
\end{figure}

\begin{figure}\centering
  \renewcommand{\arraystretch}{1.1}
  \begin{tabular}{|c|l|}\hline
    RefineWeak(2,2,2) & $n_d = 1$ ($\{s_3,\dots,e_0\}$)
    \\ \hline
    RefineWeak(2,2,3) & $n_{w,2} = 2$ ($\{s_2,\dots,e_0\}$,
    $\{s_3,\dots,e_1\}$)
    \\ \hline
    RefineWeak(2,2,4) & $n_s = 3$ ($\{s_1,\dots,e_0\}$,
    $\{s_2,\dots,e_1\}$, $\{s_3,\dots,e_2\}$)
    \\ \hline
    RefineWeak(2,2,$5\leq r\leq 8$) &
    $n_s = 9-r$ ($\{s_0,\dots,e_{r-5}\}$,\dots,$\{s_{8-r},\dots,e_{3}\}$)
    \\ \hline
  \end{tabular}
  \caption{%
    We list the RefineWeak($d$,$j$,$r$) tables used in Enumerate
    (\algref{enumerate}) for $d=2$ and $j=2$ as an example.  The start points
    and end points refer to \figref{segmentrecurse} (middle).  Unlisted values
    are zero.%
  }%
  \figlab{refineweak}
\end{figure}

In \figref{enumerate}, we use Enumerate to calculate the fraction of connected
segments $\phi_{d,L,l}$ for $d=2$ and $d=3$ for large values of $L$.  We
observe that $\phi_{d,L,l}$ tends to vary between $1/2$ and
$1/(2^d-1)$.
(We had proved in \secref{prooflowerbound} that the latter is indeed a lower
bound.)

\begin{figure}\centering
  \begin{minipage}{0.49\textwidth}\centering
    \inputtikz{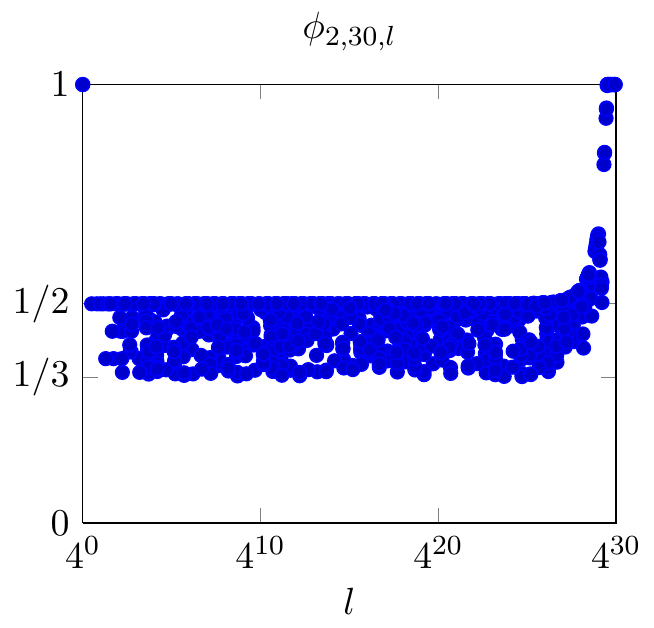}
  \end{minipage}
  \begin{minipage}{0.49\textwidth}\centering
    \inputtikz{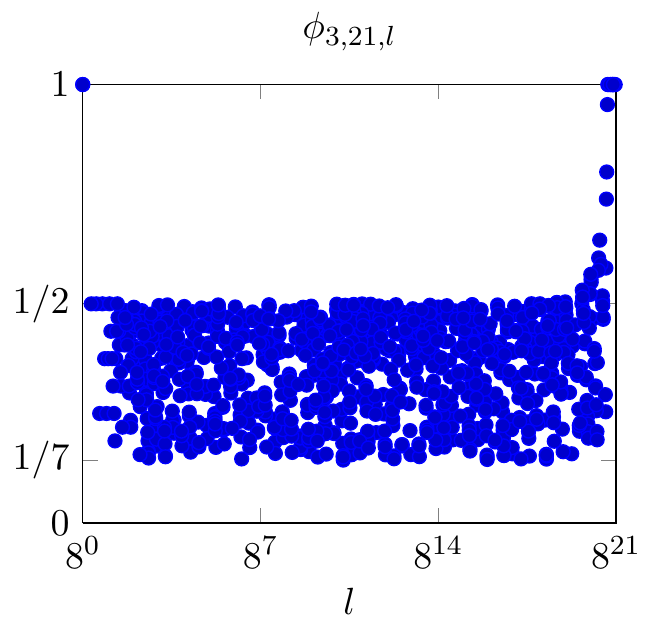}
  \end{minipage}
  \caption{%
    We plot the fraction of continuous segments of length $l$,
    $\phi_{d,L,l}$, for $d=2$ and $L=30$, (left) and $d=3$ and $L=21$
    (right), for one thousand log-uniformly randomly sampled lengths.
  }%
  \figlab{enumerate}
\end{figure}

\subsection{Computational studies---simplicial Morton curve}

For simplices, we enumerate all possible SFC segments for a given uniform
refinement level and compute the number of their face-connected components.
We achieve this by performing a depth-first search on the connectivity
graph of the submesh generated by the segment%
\footnote{\url{https://github.com/holke/sfc_conncomp}}.

While we investigate the fraction of connected SFC segments of a given
particular length $l$ among all segments of length $l$ in \secref{numerical},
here we we compute the fraction of connected segments of any length among all
possible segments.
More precisely, we compute for each possible count of connected components
the chance that any randomly choosen SFC segment (with a random length) has
exactly this number of connected components.

For a uniform level 5 refined tetrahedron we obtain that 61\% of all
SFC segments are connected and only 7\% do have four or more connected
components.
For a uniform level 8 refined triangle, about 64\% of the segments are connected
with 2\% of the segments having four or more components.
For cubes and quadrilaterals, the respective ratios of connected segments are
60\% and 71\% (here we know that the disconnected segments have exactly two
components).
We collect these results in Figure~\ref{fig:tritetnumcomp} and
Table~\ref{tab:componentpercent}.

%
%
%
%
%
%
%
%
%
%
%

\begin{figure}
        \begin{center}
\includegraphics[width=0.48\textwidth]{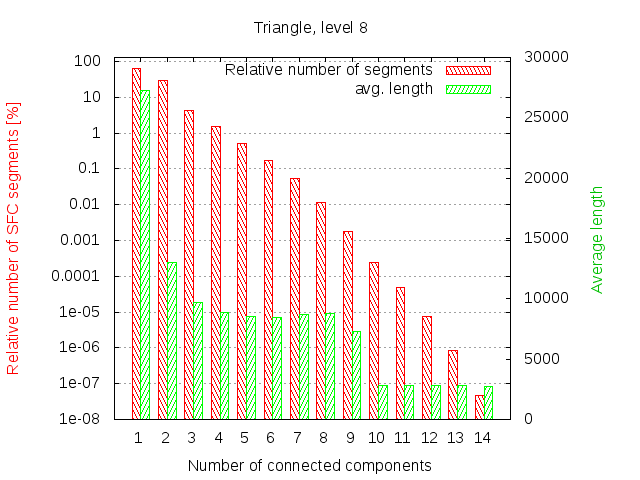}
\hfill
\includegraphics[width=0.48\textwidth]{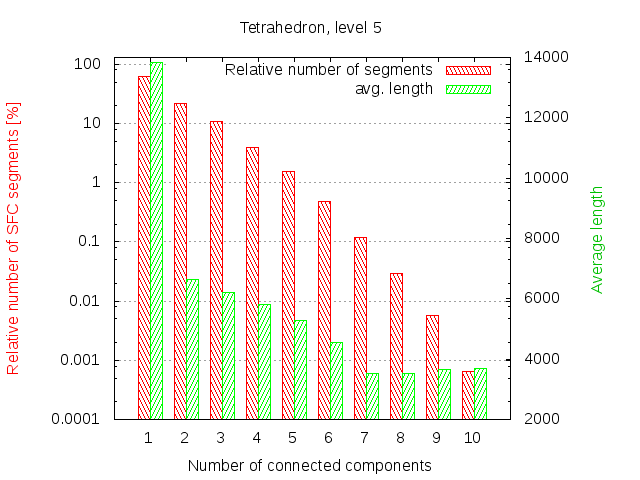}
        \end{center}
\caption{The relative count of SFC segments by number of connected components and
the average length (right y-axis) of these segments.
We exclude the segments of length 1.
Left: the distribution for a uniform level 8 refined triangle. We observe that
almost 98\% of all SFC segments have three connected components or less.
63.7\% are connected, 29.7\% have two connected components and 4.4\% have
three connected components.
Right: the distribution for a uniform level 5 refined tetrahedron.
Here, more than 93\% of the segments have three connected components or less with
61.0\% having exactly one connected component, 22.1\% with two connected 
components and 10.7\% with three connected components.
The
highest number of segments occuring agrees with
Proposition~\ref{prop:faceconncomp} (2D) and
Conjecture~\ref{con:3TMconjecture} (3D).}
\label{fig:tritetnumcomp}
\end{figure}

\begin{table}
        \begin{center}
\begin{tabular}{|l|r|r|r|r||r|r|}
\hline
&\multicolumn{4}{|c||}{Level 5} & \multicolumn{2}{c|}{Level 8}\\
\hline
& Quads & Cubes & Triangles &
\multicolumn{1}{c||}{Tets} & Quads & Triangles\\ \hline
Connected     &71.5\% & 60.0\% &63.8\% & 61.0\% &71.4\%&63.7\%\\
Non-connected &28.5\% & 40.0\% &36.2\% & 39.0\% &28.6\%&36.3\%\\
\hline
\end{tabular}
        \end{center}
\caption{The relative counts of connected and non-connected segments across all
possible SFC segments of a uniform level 5 and level 8 (2D only) refinement,
excluding segments of length one.
Here we average over a uniform distribution of lengths (in contrast to
the log-uniform distribution used in \figref{enumerate}).%
}
\label{tab:componentpercent}
\end{table}


\section{Conclusion}
\seclab{conclusion}

We prove in this document that the classical Morton or $z$-curve does not lead
to a fragmentation of the root cube into more than two disconnected subdomains.
Its loss of continuity in comparison to the Hilbert curve is thus controlled.
This is in line with experimental results
that establish the suitability of the Morton curve for
numerical applications.

We show that the bound for the recently proposed tetrahedral Morton cube is of
order $L$ and thus growing with the level of refinement.
Yet, we can demonstrate numerically that the fraction of connected to
non-connected segments is close to the cubical case.
In practice, we may expect both approaches to behave similarly.

Our result would appear relevant to make informed choices about the
type of space filling curve to use, for example in writing a new element-based
parallel code for the numerical solution of partial differential equations, or
any other code that benefits from a recursive subdivision of space.
Our theory and experiments support the existing numerical evidence that a
fragmentation of the parallel partition is not observed.


\section*{Acknowledgements}
\seclab{ack}

B.\  would like to thank Andreas Dedner for the invitation to the ICMS workshop on
Galerkin methods with applications in weather and climate forecasting, which 
provided motivation to get going proving this conjecture.
The authors would like to thank Michael Bader and Herman Haverkort for
suggesting additional relevant literature.
B.\ and H.\ acknowledge travel support by the Hausdorff Center for Mathematics
(HCM) at Bonn University
funded by the German Research Foundation (DFG).
I.\ gratefully acknowledges the support of the Intel Parallel Computing Center
at the University of Chicago.
H.\ gratefully acknowledges the scholarship support by the Bonn International
Graduate School for Mathematics (BIGS) as part of HCM.

\bibliographystyle{siam}
\bibliography{ccgo,group}

\begin{thebibliography}{10}

\bibitem{AhimianLashukVeerapaneniEtAl10}
{\sc A.~Ahimian, I.~Lashuk, S.~Veerapaneni, C.~Aparna, D.~Malhotra, I.~Moon,
  R.~Sampath, A.~Shringarpure, J.~Vetter, R.~Vuduc, D.~Zorin, and G.~Biros},
  {\em Petascale direct numerical simulation of blood flow on 200k cores and
  heterogeneous architectures}, in SC10: Proceedings of the International
  Conference for High Performance Computing, Networking, Storage, and Analysis,
  ACM/IEEE, 2010.
\newblock (Gordon Bell Prize).

\bibitem{AkcelikBielakBirosEtAl03}
{\sc V.~Ak\c{c}elik, J.~Bielak, G.~Biros, I.~Epanomeritakis, A.~Fernandez,
  O.~Ghattas, E.~J. Kim, J.~Lopez, D.~R. O'Hallaron, T.~Tu, and J.~Urbanic},
  {\em High resolution forward and inverse earthquake modeling on terascale
  computers}, in SC03: Proceedings of the International Conference for High
  Performance Computing, Networking, Storage, and Analysis, ACM/IEEE, 2003.
\newblock {G}ordon {B}ell {P}rize for {S}pecial {A}chievement.

\bibitem{Bader12}
{\sc M.~Bader}, {\em Space-Filling Curves: An Introduction with Applications in
  Scientific Computing}, Texts in Computational Science and Engineering,
  Springer, 2012.

\bibitem{BangerthHartmannKanschat07}
{\sc W.~Bangerth, R.~Hartmann, and G.~Kanschat}, {\em deal.{II} -- a
  general-purpose object-oriented finite element library}, ACM Transactions on
  Mathematical Software, 33 (2007), p.~24.

\bibitem{Bey92}
{\sc J.~Bey}, {\em Der {BPX}-{V}orkonditionierer in drei {D}imensionen:
  {G}itterverfeinerung, {P}arallelisierung und {S}imulation}, Universit{\"a}t
  Heidelberg,  (1992).
\newblock Preprint.

\bibitem{BursteddeGhattasGurnisEtAl10}
{\sc C.~Burstedde, O.~Ghattas, M.~Gurnis, T.~Isaac, G.~Stadler, T.~Warburton,
  and L.~C. Wilcox}, {\em Extreme-scale {AMR}}, in SC10: Proceedings of the
  International Conference for High Performance Computing, Networking, Storage
  and Analysis, ACM/IEEE, 2010.

\bibitem{BursteddeHolke16}
{\sc C.~Burstedde and J.~Holke}, {\em A tetrahedral space-filling curve for
  nonconforming adaptive meshes}, SIAM Journal on Scientific Computing, 38
  (2016), pp.~C471--C503.

\bibitem{BursteddeWilcoxGhattas11}
{\sc C.~Burstedde, L.~C. Wilcox, and O.~Ghattas}, {\em {\texttt{p4est}}:
  Scalable algorithms for parallel adaptive mesh refinement on forests of
  octrees}, SIAM Journal on Scientific Computing, 33 (2011), pp.~1103--1133.

\bibitem{deBergHaverkortThiteEtAl10}
{\sc M.~de~Berg, H.~Haverkort, S.~Thite, and L.~Toma}, {\em Star-quadtrees and
  guard-quadtrees: {I/O}-efficient indexes for fat triangulations and
  low-density planar subdivisions}, Computational Geometry, 43 (2010),
  pp.~493--513.

\bibitem{FinkelBentley74}
{\sc R.~A. Finkel and J.~L. Bentley}, {\em Quad trees {A} data structure for
  retrieval on composite keys}, Acta Informatica, 4 (1974), pp.~1--9.

\bibitem{GriebelZumbusch99}
{\sc M.~Griebel and G.~W. Zumbusch}, {\em Parallel multigrid in an adaptive
  {PDE} solver based on hashing and space-filling curves}, Parallel Computing,
  25 (1999), pp.~827--843.

\bibitem{HaverkortWalderveen10}
{\sc H.~Haverkort and F.~van Walderveen}, {\em Locality and bounding-box
  quality of two-dimensional space-filling curves}, Computational Geometry, 43
  (2010), pp.~131--174.

\bibitem{Hilbert91}
{\sc D.~Hilbert}, {\em {\"U}ber die stetige {A}bbildung einer {L}inie auf ein
  {F}l{\"a}chenst{\"u}ck}, Mathematische Annalen, 38 (1891), pp.~459--460.

\bibitem{Lebesgue04}
{\sc H.~L. Lebesgue}, {\em Le{\c c}ons sur l'int{\'e}gration et la recherche
  des fonctions primitives}, Gauthier-Villars, 1904.

\bibitem{Meagher82}
{\sc D.~Meagher}, {\em Geometric modeling using octree encoding}, Computer
  Graphics and Image Processing, 19 (1982), pp.~129--147.

\bibitem{Morton66}
{\sc G.~M. Morton}, {\em A computer oriented geodetic data base; and a new
  technique in file sequencing}, tech. rep., IBM Ltd., 1966.

\bibitem{Peano90}
{\sc G.~Peano}, {\em Sur une courbe, qui remplit toute une aire plane}, Math.
  Ann., 36 (1890), pp.~157--160.

\bibitem{Sagan94}
{\sc H.~Sagan}, {\em Space-Filling Curves}, Springer, 1994.

\bibitem{StewartEdwards04}
{\sc J.~R. Stewart and H.~C. Edwards}, {\em A framework approach for developing
  parallel adaptive multiphysics applications}, Finite Elements in Analysis and
  Design, 40 (2004), pp.~1599--1617.

\bibitem{WeinzierlMehl11}
{\sc T.~Weinzierl and M.~Mehl}, {\em Peano---a traversal and storage scheme for
  octree-like adaptive {C}artesian multiscale grids}, SIAM Journal on
  Scientific Computing, 33 (2011), pp.~2732--2760.

\end{thebibliography}


%

\end{document}